\documentclass{IEEEtran}
\IEEEoverridecommandlockouts
\usepackage{acronym}
\acrodef{6dma}[6DMA]{six-dimensional movable antenna}
\acrodef{6g}[6G]{six-generation}
\acrodef{mimo}[MIMO]{multi-input multi-output}
\acrodef{siso}[SISO]{single-input single-output} 
\usepackage{amsmath}
\usepackage{graphicx} 
\usepackage{epstopdf}
\usepackage{amssymb}
\usepackage{amsfonts}
\usepackage{amsthm}
\usepackage{cite}
\usepackage{bm,comment}
\usepackage{algorithm}
\usepackage{algorithmic}

\usepackage{subeqnarray}
\usepackage{subfigure}
\usepackage{multicol}
\usepackage{multirow}
\usepackage{diagbox}
\usepackage{slashbox}
\usepackage{stfloats}
\usepackage{float}
\usepackage{color}
\usepackage{cases}
\usepackage{CJK}
\usepackage{bm}

\usepackage[colorlinks,linkcolor=blue,citecolor=red]{hyperref}
\usepackage{lipsum}

\newtheorem{proposition}{\bf{Proposition}}

\usepackage{geometry}
\def\BibTeX{{\rm B\kern-.05em{\sc i\kern-.025em b}\kern-.08em
    T\kern-.1667em\lower.7ex\hbox{E}\kern-.125emX}}

\begin{document}
\title{6DMA-Aided Hybrid Beamforming with Joint Antenna Position and Orientation Optimization}

\author{Yichi Zhang, Yuchen Zhang, Lipeng Zhu, Sa Xiao, Wanbin Tang, Yonina C. Eldar, \emph{Fellow, IEEE},\\ and Rui Zhang, \emph{Fellow, IEEE}

\thanks{

Yichi Zhang, Yuchen Zhang, and Wanbin Tang are with the National Key Laboratory of Wireless Communications, University of Electronic Science and Technology of China, Chengdu 611731, China (e-mail: \{yczhang,yc\_zhang\}@std.uestc.edu.cn, wbtang@uestc.edu.cn).

Lipeng Zhu is with the Department of Electrical and Computer Engineering, National University of Singapore,
Singapore 117583 (e-mail: zhulp@nus.edu.sg).

Sa Xiao is with the National Key Laboratory of
Science and Technology on Communications, University of Electronic Science
and Technology of China, Chengdu 611731, China, and also with the Kash
Institute of Electronics and Information Industry, Kash 844000, China (e-mail: xiaosa@uestc.edu.cn).

Yonina C. Eldar is with the Faculty of Mathematics and Computer Science, Weizmann Institute of Science, Rehovot 7610001, Israel (e-mail: yonina.eldar@weizmann.ac.il).

 Rui Zhang is with School of Science and Engineering, Shenzhen Research Institute of Big Data, The Chinese University of Hong Kong, Shenzhen, Guangdong 518172, China (e-mail: rzhang@cuhk.edu.cn). He is also with the Department of Electrical and Computer Engineering, National University of Singapore, Singapore 117583 (e-mail: elezhang@nus.edu.sg).
}

}
%

\maketitle

\begin{abstract}
This paper studies a sub-connected six-dimensional movable antenna (6DMA)-aided multi-user communication system. In this system, each sub-array is connected to a dedicated radio frequency chain and collectively moves and rotates as a unit within specific local regions. The movement and rotation capabilities of 6DMAs enhance design flexibility, facilitating the capture of spatial variations for improved communication performance. To fully characterize the effect of antenna position and orientation on wireless channels between the base station (BS) and users, we develop a field-response-based 6DMA channel model to account for the antenna radiation pattern and polarization. We then maximize the sum rate of multiple users, by jointly optimizing the digital and unit-modulus analog beamformers given the transmit power budget as well as the positions and orientations of sub-arrays within given movable and rotatable ranges at the BS. Due to the highly coupled variables, the formulated optimization problem is non-convex and thus challenging to solve. We develop a fractional programming-aided alternating optimization framework  that integrates the Lagrange multiplier method, manifold optimization, and gradient descent to solve the problem. Numerical results demonstrate that the proposed 6DMA-aided sub-connected structure achieves a substantial sum-rate improvement over various benchmark schemes with less flexibility in antenna movement and can even outperform fully-digital beamforming systems that employ antenna position or orientation adjustments only. The results also highlight the necessity of considering antenna polarization for optimally adjusting antenna orientation.
\end{abstract}

\section{Introduction}
As wireless applications continue to evolve to support a substantial number of users/devices, the scarcity of spectrum resources is becoming an increasingly critical challenge. The exploration of millimeter-wave and terahertz bands for larger bandwidth and increased data rate in upcoming \ac{6g} mobile systems is critical. To alleviate the considerable path loss experienced in high-frequency bands and further enhance spectral efficiency, large-scale antenna arrays are usually employed at the base stations (BSs). This enables the simultaneous multiplexing of multiple data streams within the same frequency band, while also facilitating multi-user communication with high beamforming gains by leveraging the abundant spatial degrees of freedom (DoFs). Unfortunately, traditional fully digital beamforming requires a dedicated radio frequency (RF) chain for each antenna, resulting in heavy power consumption and high hardware cost. 

To facilitate the deployment of large-scale antenna arrays with cost-effective hardware, hybrid beamforming has been introduced as a promising technique that only requires a limited number of RF chains. The data streams are first precoded using a low-dimensional digital beamformer (DBF), followed by the application of a high-dimensional analog beamformer (ABF) composed of low-cost phase shifters (PSs) \cite{HB1,HBF1,HBF2,HBF3,HBF4,HBF5,HBF6}. There are two main hybrid beamforming architectures depending on how the RF chains are connected to the PSs \cite{HBF1,HBF2}. In the fully-connected structure, each RF chain is connected to all antennas via separate PSs. However, the extensive deployment of PSs entails a high cost. The sub-connected structure divides large-scale antenna arrays into multiple sub-arrays by connecting each RF chain exclusively with a portion of antennas of each sub-array, resulting in lower hardware cost and design complexity, but at the cost of reduced design DoFs.

To improve the communication performance with a limited number of RF chains, the authors in \cite{HBF3,HBF6} demonstrated that for a fully-connected structure, a hybrid beamformer can achieve perfect approximation of the fully-digital beamformer when the number of RF chains is no less than twice the number of data streams. Consequently, the heuristic algorithm that alternatively optimizes the DBF and ABF to approximate the fully-digital beamformer has been considered an effective method and widely used in the majority of existing works on hybrid beamforming design \cite{HBF1,HBF2,HBF3,HBF6}. However, both existing fully-connected and sub-connected hybrid beamforming schemes rely on fixed antennas (FAs), thus limiting the communication performance due to the inability to fully exploit the spatial variation of wireless channels at the transmitters/receivers.

To further utilize the spatial DoFs, some regular/irregular array geometries were explored with the objective of yielding a channel with equal singular values, thereby enhancing spatial multiplexing performance \cite{Regular_array1,Irregular_array1}. The potential of mechanically rotating the uniform linear array was also investigated to achieve better channel conditions \cite{Rotate_array}. However, the aforementioned schemes still require specific array geometries, limiting their ability to fully exploit channel spatial variation. More recently, movable antenna (MA), fluid antenna system, and other related techniques, collectively encompassing the ability of reconfiguring wireless channels via mechanically/electronically moving the antennas, have been proposed as promising technologies for improving wireless communication performance \cite{MA_history,FAS1,MA_overview,MA_implement1,MA_implement2,MA_implement3,MA1}.

Most prior works investigated MA-aided wireless systems with flexible antenna positions \cite{MA1,MA2,MA3,MA_self}. The field-response based channel model was initially presented in \cite{MA1} for an MA-aided \ac{siso} communications, revealing its performance gains over the FA system with/without antenna selection under both deterministic and stochastic channel setups. Based on this channel model, the authors in \cite{MA2} characterized the capacity of an MA-aided \ac{mimo} system, which is equipped with the MAs at both the transmitter and receiver. It was demonstrated that the positions of the transmit and receive MAs can be jointly optimized to balance the channel gain for each of the spatial data streams, thus enhancing \ac{mimo} channel capacity. The MA was generalized to multi-user systems in \cite{MA3} for the uplink transmission. Two low-complexity algorithms based on zero-forcing and minimum mean square error combining methods were introduced, demonstrating the potential of significantly reducing the total transmit power in comparison to the FA system. To balance hardware cost and communication performance, the authors in \cite{MA_self} proposed an MA-aided hybrid beamforming scheme for downlink transmission. They designed an algorithm to optimize sum rate by jointly optimizing the DBF, the ABF, and the positions of subarrays. Simulations indicated that, under certain practical scenarios, the MA-aided sub-connected structure can outperform the FA-aided fully-connected structure. In addition, MA arrays have been applied to a variety of applications, including flexible beamforming, satellite communications, unmanned aerial vehicle (UAV) communications, wireless power transfer, spectrum sharing, over-the-air computation, and physical layer security \cite{MA4,MA_Flexible_Precoding,MA_Two_Time,MA_Flexible_BF,MA_satellite,MA_UAV,MA_spectrum,MA_over_the_air,MA_PLS}. These works demonstrated that MAs provide flexible design DoFs and significantly enhance the performance of various wireless systems.

In addition to the adjustment of antenna positions, a general architecture of  \ac{6dma} was proposed in \cite{6dma1}, which can adjust three-dimensional (3D) positions and 3D rotations of antennas for adapting to the dynamic user distribution in wireless networks. It was demonstrated that \ac{6dma} can significantly improve the wireless network capacity for scenarios with non-uniform user spatial distribution, such as those with higher user/device densities in aerial/terrestrial hot-spot areas. {{Considering practical issues associated with antenna movement, such as mechanical control accuracy, CSI acquisition timeliness, and the difficulty of obtaining perfect CSI for each user, the authors in \cite{6dma2} proposed both offline and online algorithms for jointly optimizing \acp{6dma}' positions and rotations. These algorithms adjust the antenna positions and rotations discretely in scenarios where the user distribution is either known or unknown {\it a priori}.}} The authors in \cite{MA_implement1,MA_implement2,MA_implement3,6dma3} introduced various implementation architectures and practical considerations for \ac{6dma}/flexible antennas, including hardware designs and potential applications. In addition, the authors in \cite{6dma4,6dma5,6dma_UAV} further investigated specific issues in \ac{6dma}-assisted wireless communication networks, including distributed channel estimation, hybrid architectures combining fixed and movable antennas, and UAV communications.

 Most previous works on \ac{6dma} \cite{MA_implement1,MA_implement2,MA_implement3,6dma1,6dma2,6dma3,6dma4,6dma5} assumed the fully-digital structure, where each antenna is connected to a dedicated RF chain based on the fully digital beamforming structure. Additionally, they neglected the impact of rotation on antenna polarization, which can significantly affect the rate performance between the BS and users in practice. In light of this, a polarization-aware channel model was introduced for MA in \cite{6dma_polar}, which, however, this work considered the simple line-of-sight (LoS) channel only. Therefore, \ac{6dma} channels considering antenna polarization and cost-effective array structures have not been explored yet in the literature, thus motivating our current work.

In this paper, we propose a hybrid beamforming scheme based on sub-connected \ac{6dma} arrays. Specifically, rather than moving/rotating each antenna independently, the proposed 6DMA array employs a collective driving mechanism. In this reconfiguration, a dedicated motor drives each sub-array, allowing the collectively movements and rotations of all antennas within the sub-array, thereby reducing the hardware cost significantly. Besides considering the directional antenna radiation patterns, antenna polarization is incorporated in \ac{6dma} channel model, fully capturing the effects of adjusting antenna orientation and polarization jointly. The main contributions of this paper are summarized as follows:

\begin{itemize}
    \item [$\bullet$] Following the implementation of directional polarized antennas at BSs in existing wireless networks \cite{3GPP}, we extend the field-response channel model to the general case by considering the effects of multi-path superposition, directional antenna radiation pattern and polarization, aiming to accurately characterize the communication performance of \ac{6dma} systems under practical setups.
    \item [$\bullet$] Leveraging the generalized \ac{6dma} channel model, we jointly optimize the hybrid beamformer together with the positions and orientations of sub-arrays to maximize the users' sum rate. To tackle the resulting non-convex optimization problem, we introduce a fractional programming (FP)-aided alternating optimization (AO) framework. This approach combines the Lagrange multiplier method, manifold optimization (MO), and gradient descent techniques to efficiently solve the formulated problem sub-optimally.
    \item [$\bullet$] Numerical results demonstrate that our proposed \ac{6dma}-aided sub-connected structure outperforms various benchmark schemes with less flexibility in antenna movement. Moreover, under certain practical conditions, the proposed scheme with sub-connected \ac{6dma} arrays can even outperform its fully-digital counterparts with less antenna movement flexibility. {{In addition, the baseline scheme that optimizes the DBF, ABF, and the orientations and positions of subarrays based on the unpolarized channel model shows significant performance loss, highlighting the importance and necessity of considering the polarized channel model.}}
\end{itemize}

The rest of this paper is organized as follows. Section II describes the system model and formulates the problem. Section III proposes the algorithm to solve the formulated problem by an FP-aided AO frameworks. Section IV presents numerical results and pertinent discussions. Finally, Section V concludes the paper.

\emph{Notations:} Vectors (lower case) and matrices (upper case) are denoted in boldface. We use $\left(\cdot\right)^T$, $\left(\cdot\right)^{\dag}$, and $\left(\cdot\right)^H$ to denote the transpose, conjugate, and conjugate transpose, respectively, $\mathbb{C}^{M\times N}$ and $\mathbb{R}^{M\times N}$ are the set of $M\times N$ dimensional complex and real matrices, and $\mathbb{E}(\cdot)$ is the mathematical expectation operator. We use ${\bf{I}}_N$ to denote the $N\times N$ identity matrix, ${\bf{1}}_N$ is the $N$-dimensional column  vector with all one elements, ${\bf{0}}_{M\times N}$ is the $M\times N$ matrix with all zero elements, ${\bf{e}}_n\in\mathbb{R}^{N\times1}$ is the $N$-dimensional column vector where the $n$-th element is $1$ and all others are $0$, ${\bf{A}}={\rm{blkdiag}}[{\bf{a}}_1,{\bf{a}}_2,\cdots,{\bf{a}}_N]$ denotes a block diagonal matrix, ${\rm{diag}}({\bf{a}})$ is a diagonal matrix with diagonal elements being the vector ${\bf{a}}$. The notation $[a,b]\times[c,d]\times[e,f]$ represents the Cartesian product of three intervals, constituting a set of a three-dimensional vectors $(x,y,z)^T$, where $x\in[a,b]$, $y\in[c,d]$, and $z\in[e,f]$, $[f(n)]_{1\leq n\leq N}$ denotes a vector with $n$-th element being $f(n)$, where $f(n)$ is a function with respect to (w.r.t.) $n$, $[f(m,n)]_{1\leq n\leq N,1\leq m\leq M}$ is a matrix with elements at $m$-th row and $n$-th column being $f(m,n)$, where $f(m,n)$ is a function w.r.t. $m$ and $n$, $[{\bf{a}}]_n$ denotes the $n$-th element of vector ${\bf{a}}$, $||{\bf{a}}||_2$ and $||{\bf{A}}||_F$ are the $2$-norm of vector ${\bf{a}}$ and the Frobenius norm of matrix ${\bf{A}}$, respectively. We use $\rm{arcsin}(\cdot)$ to denote the arcsine function, $|b|$ and $\angle(b)$ are amplitude and phase of the complex number $b$, $\otimes$ and $\odot$ denote the Kronecker and Hadamard products, and ${\rm{Re}}\{b\}$ denotes the real part of $b$.

\section{System Model}
In this section, we first describe the \ac{6dma}-aided downlink multi-user multiple-input single-output (MU-MISO) system. Then, we extend the field-response based channel model by incorporating the effects of directional radiation pattern and polarization of antennas \cite{MA1,MA3,6dma1,6dma2,6dma3,PC_Liu,PC_Heath1,PC_Heath2,3GPP}.

\subsection{Description of the 6DMA-aided System}
We consider \ac{6dma}-aided MU-MISO communication system, consisting of a multi-antenna BS serving $K$ single-antenna users, as illustrated in Fig. \ref{system model}. The BS employs a sub-connected hybrid beamforming structure, which consists of $N$ movable uniform planar arrays (UPAs). In this structure, each movable UPA consists of $M$ antennas. The antennas within each UPA are first individually connected to a separate PS and then to a dedicated RF chain, and can move and rotate collectively in a given region. Specifically, the central point of each UPA  is constrained to move  within a specified 2D region, serving as a supporting surface. To mitigate the effects of mutual coupling between antennas, the movable regions of the UPAs are designed with sufficient separation, ensuring that the minimum distance between any two antennas is at least half of the carrier wavelength. On the other hand, to mitigate the adverse effects of mutual signal reflection, each UPA is constrained to rotate (around its central point) within a limited angular range \cite{6dma1}.

\begin{figure}
    \centering
    \includegraphics[width=0.45\textwidth]{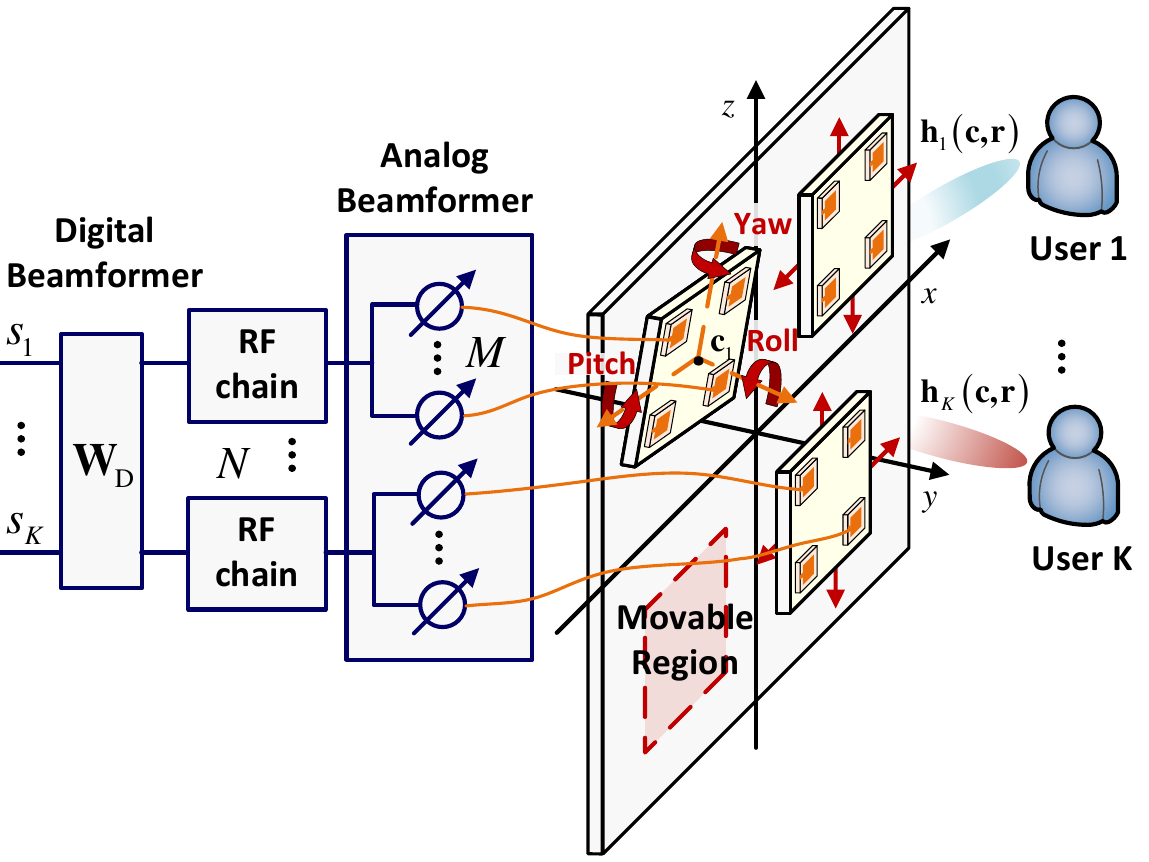}
    \caption{The proposed multi-user communication system aided by \acp{6dma} with hybrid beamforming.}
    \label{system model}
\end{figure}

The central position and orientation of the $n$-th UPA are denoted as ${\bf{c}}_n=(x_n,y_n,z_n)^T\in\mathcal{C}_n$ and ${{\bf{r}}_n=(\alpha_n,\beta_n,\gamma_n)^T\in\mathcal{R}}$, respectively, where $\mathcal{C}_n$ denotes the given movable planar region enclosed by the red-dotted line, and $\mathcal{R}$ represents the specified rotatable range. Here, $x_n$,  $y_n$, and $z_n$ represent the coordinates of the center of the $n$-th UPA along the $x$-axis, $y$-axis, and $z$-axis, respectively, within the global Cartesian coordinate system (GCS). In addition, $\alpha_n$, $\beta_n$, and $\gamma_n$ denote the pitch, roll, and yaw angles of the rotation of the $n$-th UPA around the $x$-axis, $y$-axis, and $z$-axis, respectively. The orange dashed axes in Fig. \ref{system model} indicate the local coordinate system (LCS) with arbitrary orientation associated with the $n$-th UPA. The coordinates along the $x$-axis, $y$-axis, and $z$-axis in the LCS of the $n$-th UPA are denoted by $\tilde{x}_n$, $\tilde{y}_n$, and $\tilde{z}_n$ axes, respectively.


We proceed to express the exact coordinate of each antenna within the $n$-th UPA. Let ${\bf{R}}_x$, ${\bf{R}}_y$, and ${\bf{R}}_z$ denote the coordinate transformation induced by rotation around the $x$-axis, $y$-axis, and $z$-axis, respectively. Specifically, the transformation from the LCS to the GCS is expressed by the rotation matrix as \cite{6dma1}

\vspace{-0.5cm}
\begin{equation}
\begin{aligned}
    &{\bf{R}}_n = {\bf{R}}_z\left(\gamma_n\right){\bf{R}}_y\left(\beta_n\right){\bf{R}}_x\left(\alpha_n\right)\\
    &= \begin{bmatrix}
        c_{\gamma_n} & -s_{\gamma_n} & 0\\
        s_{\gamma_n} & c_{\gamma_n} & 0\\
        0 & 0 & 1
    \end{bmatrix}
    \begin{bmatrix}
        c_{\beta_n} & 0 & s_{\beta_n}\\
        0 & 1 & 0\\
        -s_{\beta_n} & 0 & c_{\beta_n} 
    \end{bmatrix}
    \begin{bmatrix}
        1 & 0 & 0\\
        0 & c_{\alpha_n} & -s_{\alpha_n}\\
        0 & s_{\alpha_n} & c_{\alpha_n}
    \end{bmatrix}\\
    &= \begin{bmatrix}
        c_{\beta_n}c_{\gamma_n} & s_{\alpha_n}s_{\beta_n}c_{\gamma_n}-c_{\alpha_n}s_{\gamma_n} & c_{\alpha_n}s_{\beta_n}c_{\gamma_n}+s_{\alpha_n}s_{\gamma_n} \\
        c_{\beta_n}s_{\gamma_n} & s_{\alpha_n}s_{\beta_n}s_{\gamma_n}+c_{\alpha_n}c_{\gamma_n} & c_{\alpha_n}s_{\beta_n}s_{\gamma_n}-s_{\alpha_n}c_{\gamma_n} \\
         -s_{\beta_n} & s_{\alpha_n}c_{\beta_n} & c_{\alpha_n}c_{\beta_n}
    \end{bmatrix},
\end{aligned}
\end{equation}
where $c_x=\cos x$ and $s_x=\sin x$ for concise. As shown in Fig. \ref{system model}, the central point of each UPA serves as the original point of the LCS of this UPA. For each UPA, the relative positions of all antennas within it are defined as the coordinate in the UPA’s LCS and are collected in a predefined matrix ${\bm{\delta}}=[{\bm{\delta}}_1,{\bm{\delta}}_2,\cdots,{\bm{\delta}}_M]\in\mathbb{R}^{3\times M}$. The position of the $m$-th antenna within the $n$-th UPA in the GCS is given by

\vspace{-0.4cm}
\begin{equation}
    {\bf{t}}_{m,n} = {\bf{R}}_n{\bm{\delta}}_m+{\bf{c}}_n.
\end{equation}
 The MU-MISO channel gain between the transmitter and the receiver is generally a function of ${\bf{c}}_n$ and ${\bf{r}}_n$. We collect the central positions and orientations of all UPAs in 
${\bf{c}}=[{\bf{c}}_1,{\bf{c}}_2,\cdots,{\bf{c}}_N]$ and ${\bf{r}}=[{\bf{r}}_1,{\bf{r}}_2,\cdots,{\bf{r}}_N]$, respectively. The channel gain from the BS to the $k$-th user is denoted as ${\bf{h}}_k({\bf{c}},{\bf{r}})\in\mathbb{C}^{MN\times 1}$.

Let ${\bf{s}}\in \mathbb{C}^K$ be the transmit signal with covariance matrix ${\bf{I}}_K$. The $k$-th user's received signal is then given by
\vspace{-0.2cm}
\begin{equation}\label{received_signal}
\begin{aligned}
y_k&={\bf{h}}_{k}\left({\bf{c}},{\bf{r}}\right)^H{\bf{W}}_{\rm{A}}{\bf{W}}_{\rm{D}}{\bf{s}}+n_k,\\
&=\underbrace{{\bf{h}}_{k}\left({\bf{c}},{\bf{r}}\right)^H{\bf{W}}_{\rm{A}}{{\bf{w}}_{k}s_{k}}}_{\rm{Desired ~signal}}+\underbrace{{\bf{h}}_{k}\left({\bf{c}},{\bf{r}}\right)^H{\bf{W}}_{\rm{A}}\sum_{k'\neq k}{{\bf{w}}_{k'}s_{k'}}}_{\rm{Interference}}+n_k,
\end{aligned}
\end{equation}
where ${\bf{W}}_{\rm{D}}=[{\bf{w}}_1,{\bf{w}}_2,\cdots,{\bf{w}}_{K}]\in\mathbb{C}^{N\times K}$ denotes the DBF, ${\bf{W}}_{\rm{A}}\in\mathbb{C}^{MN\times N}$ denotes the ABF, and $n_k\sim\mathcal{CN}(0,\sigma_k^2)$. The sub-connected structure introduces the unit-modulus constraint to ABF \cite{HB1}, i.e.,

\begin{equation}\label{unit_modulus_constraint}
    {\bf{W}}_{\rm{A}}={\rm{blkdiag}}\left({\bf{b}}_{1},{\bf{b}}_{2},\cdots,{\bf{b}}_{N}\right).
\end{equation}
Here, ${{\bf{b}}}_{n}=\left[e^{j\psi_{m,n}}\right]^T_{1\leq m\leq M}$ represents the phase shift caused by PSs connected to the $n$-th RF chain and the $m$-th antenna, where $\psi_{m,n}\in[0,2\pi)$. Next, we derive the specific expression of channels between the BS and all users.

\subsection{Channel Model}

We consider far-field, quasi-static, and slow-fading channels \cite{MA1}. For the \ac{6dma}-aided sub-connected structure, the channel response is determined not only by the propagation environment, but also by the positions and orientations of all UPAs. Given that the size of the movable region for each UPA is considerably smaller than the signal propagation distance between the BS and the user, the far-field assumption is considered. In this scenario, the angle of departure (AoD), the angle of arrival (AoA), and the amplitude of the complex path coefficient remain unchanged during the movement of the UPAs, while only the phases vary \cite{MA1,MA2,MA3}. Furthermore, most existing studies assume that antenna elements exhibit an omnidirectional radiation pattern, overlooking spatial variations in the radiation power and the impact of polarization due to differing antenna orientations \cite{MA1,MA2,MA3,MA_self,6dma1,6dma2,6dma3}. Motivated by practical requirement, we propose a comprehensive model that incorporates the multi-path propagation, directional radiation pattern, and the effects of polarization matching. Specifically, during the rotation of the UPAs, the variations in AoD (in its LCS) result in changes in both radiation power and polarization losses, which will be detailed later. 

\subsubsection{Basic Field-Response Model}
Denote the total number of transmit and receive channel paths from the BS to the $k$-th user as $L_k^t$ and $L_k^r$, respectively. Due to the unchanged field-response vector (FRV) of the single receive FA, the paths in the rest of paper refer to the transmit channel paths. As shown in Fig. \ref{polarization model} (a), in the spherical coordinate system, the unit vector along the propagation direction is given by ${\bm{\rho}}_{k,l}=[\cos\theta_{k,l}\cos\phi_{k,l},\cos\theta_{k,l}\sin\phi_{k,l},\sin\theta_{k,l}]^T$, where $\theta_{k,l}\in[-\frac{\pi}{2},\frac{\pi}{2}]$ and $\phi_{k,l}\in[0,\pi]$ represent the elevation and azimuth angles of the $l$-th path for the $k$-th user, respectively. According to the basic geometry, the difference of the signal propagation distance between position ${\bf{t}}_{m,n}$ and reference point $[0,0,0]^T$ is given by $({\bf{t}}_{m,n})^T{\bm{\rho}}_{k,l}$. The FRV of the $m$-th antenna within $n$-th UPA for the $k$-th user can be defined as \cite{MA1,MA2,MA3}

\vspace{-0.4cm}
\begin{equation}
{\bf{g}}_k\left({\bf{t}}_{m,n}\right)=\left[e^{j\frac{2\pi}{\lambda}\left({\bf{t}}_{m,n}\right)^T{\bm{\rho}}_{k,l}}\right]^T_{1\leq l\leq L_k^t},
\end{equation}
where $\lambda$ denotes the carrier wavelength. Define the path-response matrix (PRM) ${\bm{\Sigma}}_k\in\mathbb{C}^{L_k^t\times L_k^r}$ as the response from the origin of the transmit region to that of the $k$-th user, incorporating large-scale fading effects such as path loss, and small-scale fading effects such as the random phase of each path caused by reflection, diffraction, or scattering. Specifically, the entry at the $i$-th row and $j$-th column of ${\bm{\Sigma}}_k$, denoted as $\sigma_k^{(i,j)}$, is the response coefficient between the $j$-th transmit path and the $i$-th receive path. By stacking FRVs of all \acp{6dma} as ${\bf{G}}_k({\bf{t}})=[{\bf{g}}_k({\bf{t}}_{1,1}),{\bf{g}}_k({\bf{t}}_{1,2}),\cdots,{\bf{g}}_k({\bf{t}}_{M,N})]\in\mathbb{C}^{L_k^t\times MN}$, the channel gain between the BS and the $k$-th user is given by

\begin{equation}\label{array model}
    {\bf{h}}_k\left({\bf{c}},{\bf{r}}\right)={\bf{G}}_k\left({\bf{t}}\right)^H{\bm{\Sigma}}_k{\bf{1}}_{L_k^r}.
\end{equation}

\subsubsection{Directional Radiation Pattern}
Instead of considering the omnidirection antenna with a constant pattern radiation power, we derive the received power gain based on antennas with direction radiation pattern, which are widely adopted in practice \cite{AP1,cosine_pattern1,cosine_pattern2,3GPP}. We consider the commonly used cosine radiation pattern, where the radiation power is expressed as \cite{MA_implement2,cosine_pattern1,cosine_pattern2}

\begin{figure*}[t!]
    \centering
    \includegraphics[width=0.8\textwidth]{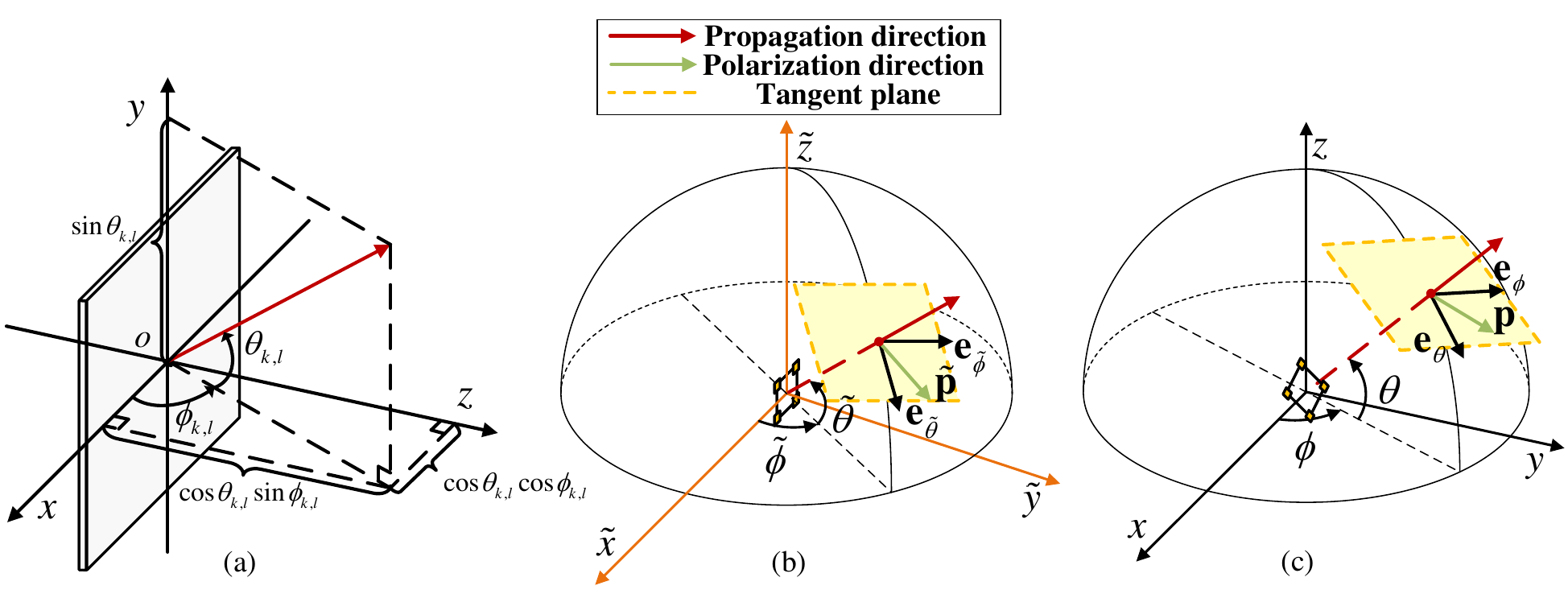}
    \caption{Illustrations of the spherical coordinate system: (a) The transformation from the Cartesian to the spherical coordinates; (b) LCS of an UPA; (c) GCS.}
    \label{polarization model}
\end{figure*}

\begin{equation}\label{antenna aperture}
    G_{E}\left(\tilde{\theta}_{k,l,n},\tilde{\phi}_{k,l,n}\right) = G_{\kappa}\left|\cos^{\kappa}\left(\tilde{\theta}_{k,l,n}\right)\right|.
\end{equation}
Here, $\tilde{\theta}_{k,l,n}\in[-\frac{\pi}{2},\frac{\pi}{2})$ and $\tilde{\phi}_{k,l,n}\in[0,2\pi]$ denote the elevation and azimuth angles of the $l$-th path for the $k$-th user in the $n$-th UPA's LCS, respectively. $\kappa> 0$ characterizes the energy focusing capability of the antenna, and $G_{\kappa}$ is a normalization factor such that $\int_{\Omega}\frac{1}{4\pi}G_{E}(\tilde{\theta}_{k,l,n},\tilde{\phi}_{k,l,n}){\rm{d}}\Omega=1$, with $\Omega$ representing the unit spherical space. For convenience, we set $\kappa=2$, $G_{\kappa}=\frac{3}{2}$, and define the pattern coefficient $A_E(\tilde{\theta}_{k,l,n})=\sqrt{G_{E}(\tilde{\theta}_{k,l,n},\tilde{\phi}_{k,l,n})}$ \cite{cosine_pattern1,cosine_pattern2}. Additionally, according to geometrical transformation, $(\tilde{\theta}_{k,l,n},\tilde{\phi}_{k,l,n})$ can be expressed as \cite{3GPP}

\vspace{-0.4cm}
\begin{subequations}
    \begin{align}
        \tilde{\theta}_{k,l,n} &= \arcsin\left(
        [0,0,1]{\bf{R}}^{-1}_n{\bm{\rho}}_{k,l}\right),\label{theta_LCS}\\
        \tilde{\phi}_{k,l,n} &= \angle\left(
        [1,j,0]{\bf{R}}^{-1}_n{\bm{\rho}}_{k,l}\right)\label{phi_LCS}.
    \end{align}
\end{subequations}

\subsubsection{Antenna Polarization}
 The radiated electromagnetic (EM) waves are composed of oscillating electric and magnetic fields, which are orthogonal to the propagation direction of waves. The polarization direction of a radiated EM wave typically denotes the oscillation direction of the electric field \cite{cosine_pattern1}. Previous works that utilize a 2D propagation model and describe the polarization directions of EM waves as purely vertical and horizontal \cite{PC_Heath1,PC_Heath2}. In contrast, our scenario accounts for the elevation spread, which cannot be neglected. Therefore, we extend the polarization model to three dimensions. It should be emphasized that the polarization direction of an EM wave is determined by the orientation of the antenna, which is reflected in the differences in elevation and azimuth angles between the LCS and GCS. Different radiation patterns may result in different polarization directions \cite{cosine_pattern1}. To incorporate the antenna polarization effects of all \acp{6dma}, the polarization direction of the EM waves for each \ac{6dma} needs to be transformed from its LCS to the GCS. 

As shown in Fig. \ref{polarization model}, the polarization direction of the EM wave from the $n$-th UPA along the $l$-th path to the $k$-th user lies in the tangent plane, with the propagation direction as its normal vector. The unit vector ${\bm{\rho}}_{k,l}$ along the propagation direction can also indicate the spherical coordinate of a point. In the $n$-th UPA's LCS, the other two spherical unit vectors, which represent the transformation of the spherical coordinate w.r.t $\tilde{\theta}_{k,l,n}$ and $\tilde{\phi}_{k,l,n}$, respectively, are given by \cite{3GPP}

\vspace{-0.4cm}
\begin{subequations}\label{spherical unit vector}
\begin{align}
    {\bf{e}}_{\tilde{\theta}_{k,l,n}} &= \frac{\partial\tilde{\bm{\rho}}_{k,l}}{\partial\tilde{\theta}_{k,l,n}} =  \begin{bmatrix}
        \sin\tilde{\theta}_{k,l,n}\cos\tilde{\phi}_{k,l,n}\\\sin\tilde{\theta}_{k,l,n}\sin\tilde{\phi}_{k,l,n}\\-\cos\tilde{\theta}_{k,l,n}
    \end{bmatrix},\\
    {\bf{e}}_{\tilde{\phi}_{k,l,n}}&= \frac{\partial\tilde{\bm{\rho}}_{k,l}}{\partial\tilde{\phi}_{k,l,n}} =  \begin{bmatrix}
        -\sin\tilde{\phi}_{k,l,n}\\\cos\tilde{\phi}_{k,l,n}\\0
    \end{bmatrix}.
    \end{align}
\end{subequations} The polarization direction $\tilde{\bf{p}}_{k,l,n}^t$ can be decomposed into a linear combination of the two spherical unit vectors within the tangent plane as

\begin{equation}\label{polarization vector LCS}
    \tilde{\bf{p}}_{k,l,n}^t=F_{\tilde{\theta}_{k,l,n}}{\bf{e}}_{\tilde{\theta}_{k,l,n}}+F_{\tilde{\phi}_{k,l,n}}{\bf{e}}_{\tilde{\phi}_{k,l,n}},
\end{equation} 
 where $F_{\tilde{\theta}_{k,l,n}}$ and $F_{\tilde{\phi}_{k,l,n}}$ denote the polarized field components of the $n$-th UPA along the two directions ${\bf{e}}_{\tilde{\theta}_{k,l,n}}$ and ${\bf{e}}_{\tilde{\phi}_{k,l,n}}$, respectively. For simplicity, we assume vertically polarized antennas, leading to $F_{\tilde{{\theta}}_{k,l,n}}=1$ and $F_{\tilde{{\phi}}_{k,l,n}}=0$ \cite{Polar_Jsac,3GPP}. 
 
 The expression of the polarization direction in \eqref{polarization vector LCS} is based on the elevation and azimuth angles in each UPA's LCS, complicating the process of incorporating the antenna polarization effects of all \acp{6dma} into the channel model. To address this, the polarization direction should be transformed into the GCS. As shown in Fig. \ref{polarization model} (c), in the GCS, the polarization direction can be similarly expressed as a linear combination of the other two spherical unit vectors ${\bf{e}}_{{\theta}_{k,l}}=\frac{\partial{\bm{\rho}}_{k,l}}{\partial\theta_{k,l}}$ and ${\bf{e}}_{{\phi}_{k,l}}=\frac{\partial{\bm{\rho}}_{k,l}}{\partial\phi_{k,l}}$. The two vectors follow the definitions in \eqref{spherical unit vector} with the elevation and azimuth angles defined in the GCS. Let $F_{{\theta}_{k,l,n}}$ and $F_{{\phi}_{k,l,n}}$ denote the polarized field components of the $n$-th UPA along the directions ${\bf{e}}_{{\theta}_{k,l}}$ and ${\bf{e}}_{{\phi}_{k,l}}$, respectively. The relation between the two pairs of electric field components is given by \cite{3GPP}

\begin{equation}
    \begin{bmatrix}
        F_{{\theta}_{k,l,n}}\\F_{{\phi}_{k,l,n}}
    \end{bmatrix} = \begin{bmatrix}
        {\bf{e}}_{{\theta}_{k,l}}^T{\bf{R}}_n{\bf{e}}_{\tilde{{{\theta}}}_{k,l,n}} & {\bf{e}}_{{\theta}_{k,l}}^T{\bf{R}}_n{\bf{e}}_{\tilde{{{\phi}}}_{k,l,n}}\\
        {\bf{e}}_{{\phi}_{k,l}}^T{\bf{R}}_n{\bf{e}}_{\tilde{{{\theta}}}_{k,l,n}} & {\bf{e}}_{{\phi}_{k,l}}^T{\bf{R}}_n{\bf{e}}_{\tilde{{{\phi}}}_{k,l,n}}
    \end{bmatrix}\begin{bmatrix}
        F_{\tilde{{\theta}}_{k,l,n}}\\F_{\tilde{{\phi}}_{k,l,n}}
    \end{bmatrix}.
\end{equation}
Then, the polarization direction in the GCS is given by 
\begin{equation}\label{polarization vector}
\begin{aligned}
    {\bf{p}}_{k,l,n}^t&=F_{{\theta}_{k,l,n}}{\bf{e}}_{{\theta}_{k,l}}+F_{{\phi}_{k,l,n}}{\bf{e}}_{{\phi}_{k,l}}\\
    &= \left({\bf{e}}_{{\theta}_{k,l}}{\bf{e}}_{{\theta}_{k,l}}^T+{\bf{e}}_{{\phi}_{k,l}}{\bf{e}}_{{\phi}_{k,l}}^T\right){\bf{R}}_n{\bf{e}}_{\tilde{{{\theta}}}_{k,l,n}}.
\end{aligned}
\end{equation} 
Notably, the antenna polarization effects of receive antennas should also be considered to characterize the transmit-receive polarization mismatch, which affects the power received by an antenna. However, due to the fixed orientation at receive antennas, the antenna polarization of receive antennas remains unchanged. Without loss of generality, we omit the specific expressions for cross-polarization discrimination and the polarized field components at the receiver. Instead, we define the equivalent receive polarization direction for the $l$-th path associated with the $k$-th user in the GCS as ${\bf{p}}_{k,l}^r$. The channel gain accounting for polarization mismatch is given by

\begin{equation}\label{polarization matching}
    A_{P}\left(\tilde{\theta}_{k,l,n},\tilde{\phi}_{k,l,n}\right) = \left({\bf{p}}_{k,l,n}^t\right)^T{\bf{p}}_{k,l}^r.
\end{equation}
Define ${\bf{A}}_{k}({\bf{r}}) = [A_{E}(\tilde{\theta}_{k,l,n})A_{P}(\tilde{\theta}_{k,l,n},\tilde{\phi}_{k,l,n})]_{1\leq n\leq N,1\leq l\leq L_k^t}\in\mathbb{R}^{N\times L_K^t}$. The resulting channel vector between the BS and the $k$-th user is given by

\begin{equation}\label{general channel model}
    {\bf{h}}_{k}\left({\bf{c}},{\bf{r}}\right) = \bigg({\bf{A}}_{k}\left({\bf{r}}\right)\otimes{\bf{1}}_M^T
    \bigg)\odot{\bf{G}}_k\left({\bf{t}}\right)^H{\bm{\Sigma}}_k{\bf{1}}_{L_k^r}.
\end{equation}
 Based on the derived channel model and $\eqref{received_signal}$, the achievable rate of the $k$-th user in bits-per-second-per-Hz (bps/Hz) is given by

\begin{equation}\label{rate}
R_k=\log_2\left(1+\frac{\left|{\bf{h}}_{k}\left({\bf{c}},{\bf{r}}\right)^H{\bf{W}}_{\rm{A}}{\bf{w}}_k\right|^2}{\sum_{k'\neq k}\left|{\bf{h}}_{k}\left({\bf{c}},{\bf{r}}\right)^H{\bf{W}}_{\rm{A}}{\bf{w}}_{k'}\right|^2+\sigma_k^2}\right).
\end{equation}

\subsection{Problem Formulation}
To maximize the sum rate of all users, we jointly optimize the central positions ${\bf{c}}$ and orientations ${\bf{r}}$ of all UPAs, the ABF ${\bf{W}}_{\rm{A}}$, and the DBF ${\bf{W}}_{\rm{D}}$, as formulated in

\begin{subequations}\label{optimization_problem}
\begin{align}
    \max_{{\bf{c}},{\bf{r}},{\bf{W}}_{\rm{A}},{\bf{W}}_{\rm{D}}}&\sum_{k=1}^KR_k\\
    {\rm{s.t.}}~~~~&{\bf{c}}_{n}\in\mathcal{C}_{n},~{\bf{r}}_{n}\in\mathcal{R},~\forall n,\label{position_constraint}\\
    &\left|\left|{\bf{W}}_{{\rm{A}}}{\bf{W}}_{{\rm{D}}}\right|\right|_F^2\leq P,\label{power_budget}\\
    &\left|\left[{\bf{b}}_{n}\right]_m\right|=1,~\forall m,n\label{unit_modulus},
\end{align}
\end{subequations}
where $P$ represents the power budget, ${\mathcal{C}}_n=[x_n^L,x_n^U]\times[y_n^L,y_n^U]\times[z_n^L,z_n^U]$, and $\mathcal{R}=[\alpha^L,\alpha^U]\times[\beta^L,\beta^U]\times[\gamma^L,\gamma^U]$. Here, the variables with superscripts $L$ and $U$ denote the lower and upper boundary constraints for the set of variables, respectively. Unlike the previous FA-based hybrid beamforming designs \cite{HBF1,HBF2,HBF3}, the movable positions and rotatable orientations of UPAs are highly coupled in the objective function. {{Furthermore, the prior studies that optimize the positions and orientations of \acp{6dma} using an LoS channel model with a fully-digital architecture and ignore polarization effects for antenna rotation \cite{6dma1,6dma2}. To make the scenario more practical, we reconstruct a channel model that accounts for the multi-path propagation environment, directional radiation pattern, and antenna polarization. We then consider the hybrid beamforming structure to reduce hardware cost. These emerging complexities introduce a non-concave objective function, non-convex constraints, and highly coupled variables, making \eqref{optimization_problem} challenging to solve optimally. Thus, new algorithms are required to solve problem \eqref{optimization_problem} efficiently.}}

\section{Algorithm Design}

This section introduces an FP-aided AO algorithm to solve $\eqref{optimization_problem}$ sub-optimally \cite{FP}. The FP framework first reformulates $\eqref{optimization_problem}$ into a more tractable equivalent form. Inspired by AO, the resulting FP problem is then decomposed into
four subproblems to iteratively update the DBF, the ABF,
the central positions and orientations of UPAs. In particular, we introduce a direct method that maintains a consistent objective function throughout the entire optimization workflow, different from the indirect methods in existing literature \cite{HBF1, HBF2, HBF3} that derive the DBFs and ABFs through approximation of the fully-digital beamformer. This guarantees the convergence of the FP-aided AO algorithm. Specifically, the closed-form DBF is drived using the Lagrange multiplier method. In addition, the MO is considered to handle the unit-modulus constraint in ABF optimization. A gradient descent method is also adopted to sequentially search the positions and orientations of UPAs.

\subsection{Problem Reformulation Based on FP}
To simplify the objective function, we adopt the FP framework presented in \cite{FP}. Let $\omega_k={\bf{h}}_{k}\left({\bf{c}},{\bf{r}}\right)^H{\bf{W}}_{\rm{A}}{\bf{w}}_k$ and $\nu_k=\sigma_k^2+\sum_{k'=1}^K\left|{\bf{h}}_{k}\left({\bf{c}},{\bf{r}}\right)^H{\bf{W}}_{\rm{A}}{\bf{w}}_{k'}\right|^2$, the equivalent form of problem \eqref{optimization_problem} is formulated as

\begin{subequations}\label{FP}
\begin{align}
\max_{{\bf{W}}_{\rm{D}},{\bf{W}}_{\rm{A}},{\bf{c}},{\bf{r}},{\bf{u}},{\bf{v}}}&\mathcal{L}=\sum_{k=1}^K\left(1+u_k\right)\left(2{\rm{Re}}\left\{v_k^{\dag}\omega_k\right\}-\left|v_k\right|^2\nu_k\right)\notag\\
&\quad~~+\log\left(1+u_k\right)-u_k\\
{\rm {s.t.}}~~&u_k>0,v_k\in\mathbb{C},\forall k,\\
&\eqref{position_constraint}, \eqref{power_budget}, \eqref{unit_modulus},
\end{align}
\end{subequations}
\noindent where ${\bf{u}}=[u_1,u_2,\cdots,u_K]$ and ${\bf{v}}=[v_1,v_2,\cdots,v_K]$ are the slack variables.

To decouple the variables $\{{\bf{W}}_{\rm{D}},{\bf{W}}_{\rm{A}},{\bf{c}},{\bf{r}},{\bf{u}},{\bf{v}}\}$ in $\eqref{FP}$, we employ the AO framework to update each variable iteratively with other variables being fixed. Given ${\bf{W}}_{\rm{D}}$, ${\bf{W}}_{\rm{A}}$, ${\bf{c}}$, and ${\bf{r}}$, $u_k$ and $v_k$ are optimally given by

\begin{equation}\label{gamma}
u_k^*=\frac{\left|{\bf{h}}_{k}\left({\bf{c}},{\bf{r}}\right)^H{\bf{W}}_{\rm{A}}{\bf{w}}_k\right|^2}{\sigma_k^2+\sum_{k'\neq k}\left|{\bf{h}}_{k}\left({\bf{c}},{\bf{r}}\right)^H{\bf{W}}_{\rm{A}}{\bf{w}}_{k'}\right|^2}
\end{equation}
and 
\begin{equation}\label{omega}
v_k^*=\frac{{\bf{h}}_{k}\left({\bf{c}},{\bf{r}}\right)^H{\bf{W}}_{\rm{A}}{\bf{w}}_k}{\sigma_k^2+\sum_{k'=1}^K\left|{\bf{h}}_{k}\left({\bf{c}},{\bf{r}}\right)^H{\bf{W}}_{\rm{A}}{\bf{w}}_{k'}\right|^2},
\end{equation}
\noindent respectively. Note that the objective function contains a term $\sum_{k=1}^K(\log(1+u_k)-u_k-(1+u_k)|v_k|^2\sigma_k^2)$, which is only determined by variables ${\bf{u}}$ and ${\bf{v}}$. For convenience, we omit it in the rest of the paper.

\subsection{DBF Design}
Given $\{{\bf{W}}_{\rm{A}},{\bf{c}},{\bf{r}},{\bf{u}},{\bf{v}}\}$, we aim to determine the DBF ${\bf{W}}_{\rm{D}}$. The power constraint can be rewritten as \cite{HBF1}

\begin{equation}
\left|\left|{\bf{W}}_{{\rm{A}}}{\bf{W}}_{{\rm{D}}}\right|\right|_F^2=M\left|\left|{\bf{W}}_{{\rm{D}}}\right|\right|_F^2\leq P.
\end{equation}

\noindent Thus, ${\bf{W}}_{\rm{A}}$ is irrelevant to the constraint $\eqref{power_budget}$. The DBF design subproblem is expressed as

\begin{subequations}\label{D_optimize}
\begin{align}
\max_{{\bf{W}}_{\rm{D}}}~&\sum_{k=1}^K\left(-\mu_k\sum_{k'=1}^K\left|{\bm{\psi}}_{k}^H{\bf{w}}_{k'}\right|^2+2\rm{Re}\left\{\bm{\omega}_k^H{\bf{w}}_k\right\}\right)\\
{\rm {s.t.}}~&\left|\left|{\bf{W}}_{{\rm{D}}}\right|\right|_F^2\leq \frac{P}{M},
\end{align}
\end{subequations}
\noindent where $\bm{\omega}_k^H=(1+u_k)v_k^{\dag}{\bm{\psi}}_k^H$, $\bm{\psi}_k^H={\bf{h}}_k({\bf{c}},{\bf{r}})^H{\bf{W}}_{\rm{A}}$, and $\mu_k=(1+u_k)|v_k^2|$. Let ${\bm{\Psi}}=\sum_{k=1}^K\mu_k\bm{\psi}_k\bm{\psi}_k^H$. The Lagrangian multiplier approach enables a closed-form solution for the convex quadratic optimization problem, given by

\begin{equation}\label{D_optimal}
{\bf{w}}_k^*=\left({\bm{\Psi}}+\lambda{\bf{I}}_{N}\right)^{-1}\bm{\omega}_k,
\end{equation}
where $\lambda$ represents the Lagrange multiplier that satisfies $\lambda(||{\bf{W}}_{\rm{D}}||_F^2-\frac{P}{M})=0$. If $||{\bf{W}}_{\rm{D}}^*||_F^2\leq\frac{P}{M}$, $\lambda=0$. Otherwise, we use the bisection search to determine $\lambda$.

\subsection{ABF Design}

Given $\{{\bf{W}}_{\rm{D}},{\bf{c}},{\bf{r}},{\bf{u}},{\bf{v}}\}$, the next step is to optimize ${\bf{W}}_{\rm{A}}$. First, the objective function is reconstructed to decouple the ABF from the other variables. Let $\bar{{\bf{b}}}=[{\bf{b}}_{1}^T,{\bf{b}}_{1}^T\cdots,{\bf{b}}_{N}^T]^T$, $\bar{{\bf{h}}}_{k,k'}^H={\bf{h}}_{k}({\bf{c}},{\bf{r}})^H({\rm{diag}}({\bf{w}}_{k'})\otimes {{\bf{I}}_{M}})$, and $\bar{{\bm{\omega}}}_k^H=(1+u_k)v_k^{\dag}\bar{{\bf{h}}}_{k,k}^H$. The ABF design subproblem is reformulated as

\begin{subequations}\label{A_optimize}
\begin{align}
\min_{\bar{\bf{b}}}~&\mathcal{\bar{L}}\left({\bar{\bf{b}}}\right)=\sum_{k=1}^K\left(\mu_k\sum_{k'=1}^K\left|\bar{{\bf{h}}}_{k,k'}^H{\bf{\bar{b}}}\right|^2-2\rm{Re}\left\{\bar{{\bm{\omega}}}_k^H{\bf{\bar{b}}}\right\}\right)\label{A_objective}\\
{\rm{s.t.}}&\left|\left[\bar{\bf{b}}\right]_s\right|=1,\forall s=1,2,\cdots,MN.\label{umc}
\end{align}
\end{subequations}
The main obstacle of solving the subproblem is the non-convex constraint \eqref{umc}. The unit-modulus constraint is common in reconfigurable intelligent surface-aided communications and hybrid beamfoming designs. It can be tackled through the penalty method or the majorization-maximization algorithm \cite{HBF4,RIS}. However, the results of these methods depend heavily on the choice of the initial value, which can potentially lead to worse solutions. In this subsection, we use the MO method, which has been claimed to find a near-optimal solution for the ABF design and perform well in simulations \cite{HBF1,HBF5}.

The unit-modulus constraints can be regarded as a search space (feasible region) $\mathcal{S}=\{\bar{\bf{b}}\in\mathbb{C}^{MN\times 1}:|[\bar{\bf{b}}]_s|=1,~\forall s=1,2,\cdots,MN\}$, which in \eqref{A_optimize} is the product of $MN$ complex circles, forming a Riemannian submanifold of $\mathbb{C}^{MN\times 1}$. In order to fully realize Riemannian MO, the algorithm must guarantee each iteration point remaining on the manifold and the objective function increasing with each iteration. Next, we will briefly introduce several concepts to meet these criteria.

For a given point $\bar{\bf{b}}$ on the manifold $\mathcal{S}$, to prevent excessive deviation of the update point from the manifold $\mathcal{S}$, we choose the search directions on its tangent plane, which is given by

\begin{equation}
    T_{\bar{\bf{b}}}{\mathcal{S}} = \left\{{\bf{z}}\in{\mathbb{C}}^{MN\times 1}:\rm{Re}\left\{{\bf{z}}\odot{\bar{{\bf{b}}}}^{\dag}\right\}={\bf{0}}_{MN}\right\}.
\end{equation}
The Riemannian gradient denotes the search direction of the steepest ascent of the objective function \eqref{A_objective}. It can be calculated by projecting the Euclidean gradient $\nabla_{\bar{\bf{b}}^{\dag}}\mathcal{L}$ orthogonally onto the tangent plane of the manifold $\mathcal{S}$, as given by \cite{Manifold}

\begin{figure*}[!b]
    \normalsize
    \setcounter{equation}{33}
    \hrulefill
    \begin{equation}
        \label{partial_c}
        \frac{\partial{\bf{h}}_k^H\left({\bf{c}},{\bf{r}}\right)}{\partial{\bf{c}}_n} = \left[{\bf{0}}_{3\times \left(n-1\right)M}, \left[\sum_{l=1}^{L_k^t}j\frac{2\pi}{\lambda}{\hat{\sigma}}_{k,l}^*{{\bf{e}}_3^T{\bf{R}}^T_n\left({\bm{\rho}}_{k,l}{\bm{\rho}}_{k,l}^T-{\bf{I}}_3\right){\bf{p}}_{k,l}^r}e^{j\frac{2\pi}{\lambda}\left({\bf{t}}_{m,n}\right)^T{\bm{\rho}}_{k,l}}{\bm{\rho}}_{k,l}\right]_{1\leq m\leq M},{\bf{0}}_{3\times \left(N-n\right)M}\right]
    \end{equation}

    \begin{equation}
        \label{partial_r}
        \left[\frac{\partial{\bf{h}}_k^H\left({\bf{c}},{\bf{r}}\right)}{\partial\Gamma_n}\right]_{\left(n'-1\right)M+m} = \left\{\begin{aligned}
            &\sum_{l=1}^{L_k^t}{\left(j\frac{2\pi}{\lambda}{\hat{\sigma}}_{k,l}^*{{\bf{e}}_3^T{\bf{R}}^T_{n'}\left({\bm{\rho}}_{k,l}{\bm{\rho}}_{k,l}^T-{\bf{I}}_3\right){\bf{p}}_{k,l}^r}e^{j\frac{2\pi}{\lambda}\left({\bf{t}}_{m,n'}\right)^T{\bm{\rho}}_{k,l}}{\bm{\delta}}_m^T\frac{\partial{\bf{R}}^T_{n'}}{\partial{\Gamma}_n}{\bm{\rho}}_{k,l}\right.}\\
        &\quad\quad\left.+{\hat{\sigma}}_{k,l}^*{{\bf{e}}_3^T\frac{\partial{\bf{R}}^T_{n'}}{\partial{\Gamma}_n}\left({\bm{\rho}}_{k,l}{\bm{\rho}}_{k,l}^T-{\bf{I}}_3\right){\bf{p}}_{k,l}^r}e^{j\frac{2\pi}{\lambda}\left({\bf{t}}_{m,n'}\right)^T{\bm{\rho}}_{k,l}}\right),&&n'=n,\\
        &0, && {\rm{otherwise}}.
        \end{aligned}\right.
    \end{equation}
\setcounter{equation}{24}
\end{figure*}

\begin{equation}
    {\rm{grad}}_{\bar{\bf{b}}}\mathcal{\bar{L}} = {\rm{Proj}}_{\bar{\bf{b}}}\nabla_{\bar{\bf{b}}^{\dag}}\mathcal{\bar{L}} = \nabla_{\bar{\bf{b}}^{\dag}}\mathcal{\bar{L}}-{\rm{Re}}\left\{\nabla_{\bar{\bf{b}}^{\dag}}\mathcal{\bar{L}}\odot{\bar{\bf{b}}}^{\dag}\right\}\odot{\bar{\bf{b}}},
\end{equation}
where
\begin{equation}
    \nabla_{\bar{\bf{b}}^{\dag}}\mathcal{\bar{L}}=\sum_{k=1}^K\left(\mu_k\sum_{k'=1}^K\bar{{\bf{h}}}_{k,k'}\bar{{\bf{h}}}_{k,k'}^H{\bf{\bar{b}}}-\bar{{\bm{\omega}}}_k\right).
\end{equation}
Furthermore, we follow the well-known Polak-Ribiere parameter to ensure that the objective value remains non-increasing in each iteration \cite{Manifold2}. The search direction ${\bm{\eta}}^{(i+1)}$ in the $(i+1)$-th iteration is determined by its Riemannian gradient and the projection of the $i$-th search direction ${\bm{\eta}}^{(i)}$ from the $i$-th tangent plane onto the $(i+1)$-th one, as given by

\begin{algorithm}[t!]
    \renewcommand{\algorithmicrequire}{\bf{Input:}}
    \renewcommand{\algorithmicensure}{\bf{Output:}}
    \caption{MO-Based ABF Design}
    \label{MO}
    \begin{algorithmic}[1]
    \REQUIRE${\bf{W}}_{\rm{D}},{\bf{c}},{\bf{r}},{\bf{u}},{\bf{v}},{\bar{\bf{b}}}^{(0)},\varepsilon,I_{max}$.
    \ENSURE ${\bar{\bf{b}}}^{*}$.
    \STATE Set iteration index $i=0$ and ${\bm{\eta}}^{(0)}=-{\rm{grad}}_{\bar{\bf{b}}^{\left(0\right)}}\mathcal{L}$.
    \REPEAT
    \STATE Determine the step size $v^{\left(i\right)}$ according to \cite{Manifold}.
    \STATE Obtain ${\bar{\bf{b}}^{\left(i+1\right)}}$ via $\eqref{retraction}$.
    \STATE Determine $u^{\left(i+1\right)}$ according to \cite{Manifold}.
    \STATE Obtain ${\bm{\eta}}^{\left(i+1\right)}$ via $\eqref{search direction}$.

    \STATE Set $i=i+1$.
    \UNTIL{$\left|\mathcal{\bar{L}}\left(\bar{\bf{b}}^{\left(i+1\right)}\right)-\mathcal{\bar{L}}\left(\bar{\bf{b}}^{\left(i\right)}\right)\right|<\varepsilon$} or the maximum iteration number $I_{max}$ is reached.
    \end{algorithmic}
\end{algorithm}

\begin{equation}\label{search direction}
    {\bm{\eta}}^{\left(i+1\right)} = -{\rm{grad}}_{\bar{\bf{b}}^{\left(i+1\right)}}\mathcal{\bar{L}}+u^{\left(i+1\right)}{\rm{Proj}}_{\bar{\bf{b}}^{\left(i+1\right)}}{\bm{\eta}}^{\left(i\right)},
\end{equation}
where $u$ is the Polak-Ribirer parameter and the superscripts $i$ and $i+1$ denote the variables in the corresponding iterations.

After determining the search direction, the retraction is adopted to map the next point from the tangent plane onto the manifold $\mathcal{S}$. Similarly, the Armijo backtracking line search step $v^{(i)}$ is accepted to ensure the non-increasing objective function, thus leading to

\begin{equation}\label{retraction}
    \left[{\bar{\bf{b}}^{\left(i+1\right)}}\right]_s = \frac{\left[{\bar{\bf{b}}^{\left(i\right)}}+v^{\left(i\right)}{\bm{\eta}}^{\left(i\right)}\right]_s}{\left|\left[{\bar{\bf{b}}^{\left(i\right)}}+v^{\left(i\right)}{\bm{\eta}}^{\left(i\right)}\right]_s\right|},~\forall s=1,2,\cdots,MN.
\end{equation}
The overall MO workflow is summarized as Algorithm \ref{MO}. Furthermore, the Riemannian gradient approaches zero as the update of the iteration point, thereby confirming the convergence of Algorithm \ref{MO} \cite{Manifold}.

\subsection{Antenna Position and Orientation Design}

Given $\{{\bf{W}}_{\rm{D}},{\bf{W}}_{\rm{A}},{\bf{u}},{\bf{v}}\}$, we consider sequentially optimizing the position and orientation of \ac{6dma} ${\bf{c}}_n$ and ${\bf{r}}_n$ with the other central points and orientations of UPAs being fixed. Let $\hat{\bm{\Psi}}_{k}=\mu_k\sum_{k'=1}^{K}{\bf{W}}_{\rm{A}}{\bf{w}}_{{k'}}{\bf{w}}_{{k'}}^H{\bf{W}}_{\rm{A}}^H$ and $\hat{\bm{\omega}}_{k}=(1+u_k)v_k^{\dag}{\bf{W}}_{\rm{A}}{\bf{w}}_{{k}}$. Given ${\bf{r}}_n$, the \ac{6dma} position design subproblem is simplified as

\begin{equation}\label{P_optimization}
\begin{aligned}
\max_{{\bf{c}}_{n}\in\mathcal{C}_{n}}{\mathcal{L}\left({\bf{c}}_n,{\bf{r}}_n\right)}=&\sum_{k=1}^K\left(-{\bf{h}}_{k}\left({\bf{c}},{\bf{r}}\right)^H\hat{\bf{\Psi}}_k{\bf{h}}_{k}\left({\bf{c}},{\bf{r}}\right)\right.\\
&+\left.2\rm{Re}\left\{{\bf{h}}_{k}\left({\bf{c}},{\bf{r}}\right)^H\hat{\bm{\omega}}_{k}\right\}\right),
\end{aligned}
\end{equation}

\noindent
To reduce the complexity for addressing ${\mathcal{L}({\bf{c}}_n,{\bf{r}}_n)}$, a gradient descent method is considered to update ${\bf{c}}_{n}$. Similarly, we optimize ${\bf{r}}_{n}$ following the same procedure with fixed ${\bf{c}}_n$. For $n$-th UPA, the gradient values of ${\bf{c}}_n$ and ${\bf{r}}_n$ are calculated as

\begin{equation}\label{gradient_value}
\begin{aligned}
\nabla_{{\bf{c}}_{n}}{\mathcal{L}\left({\bf{c}},{\bf{r}}\right)}=2\sum_{k=1}^K\rm{Re}\left\{\frac{\partial{{\bf{h}}_{k}^H\left({\bf{c}},{\bf{r}}\right)}}{\partial{{\bf{c}}_{n}}}\left(\hat{\bm{\omega}}_{k}-\hat{\bf{\Psi}}_k{\bf{h}}_{k}\left({\bf{c}},{\bf{r}}\right)\right)\right\},
\end{aligned}
\end{equation}
and
\begin{equation}\label{gradient_value_r}
\begin{aligned}
\nabla_{{\bf{r}}_{n}}{\mathcal{L}\left({\bf{c}},{\bf{r}}\right)}=2\sum_{k=1}^K\rm{Re}\left\{\frac{\partial{{\bf{h}}_{k}^H\left({\bf{c}},{\bf{r}}\right)}}{\partial{{\bf{r}}_{n}}}\left(\hat{\bm{\omega}}_{k}-\hat{\bf{\Psi}}_k{\bf{h}}_{k}\left({\bf{c}},{\bf{r}}\right)\right)\right\},
\end{aligned}
\end{equation}
respectively. To further derive the specific expression of the gradient, the simplified expression of the parameters in the channel model \eqref{general channel model} is presented in the following proposition.

\begin{proposition}\label{expressions of channel gains}
Combining $\eqref{antenna aperture}$, $\eqref{theta_LCS}$, \eqref{polarization vector}, \eqref{polarization matching}, and the fact that ${\bf{R}}^{-1}_n={\bf{R}}^T_n$, the power gains resulting from the antenna radiation pattern and polarization matching are given by

\begin{subequations}
\begin{align}
    A_{E}\left(\tilde{\theta}_{k,l,n}\right) &= \sqrt{\frac{3}{2}\left({1-\left({\bf{e}}_3^T{\bf{R}}^T_n{\bm{\rho}}_{k,l}\right)^2}\right)},\label{aperture gain}\\
    A_{P}\left(\tilde{\theta}_{k,l,n},\tilde{\phi}_{k,l,n}\right) &= \frac{{\bf{e}}_3^T{\bf{R}}^T_n\left({\bm{\rho}}_{k,l}{\bm{\rho}}_{k,l}^T-{\bf{I}}_3\right){\bf{p}}_{k,l}^r}{\sqrt{{1-\left({\bf{e}}_3^T{\bf{R}}^T_n{\bm{\rho}}_{k,l}\right)^2}}}.\label{polarization gain}
\end{align}
\end{subequations}
 Let $\hat{\sigma}_{k,l}=\sqrt{\frac{3}{2}}\sum_{l'=1}^{L_k^r}\sigma_k^{(l,l')}$. By substituting \eqref{aperture gain} and \eqref{polarization gain} into the general channel model \eqref{general channel model}, we have

\begin{equation}
\begin{aligned}
    &\left[{\bf{h}}_k\left({\bf{c}},{\bf{r}}\right)\right]_{\left(n-1\right)M+m} \\&=  \sum_{l=1}^{L_k^t}{{\hat{\sigma}}_{k,l}{{\bf{e}}_3^T{\bf{R}}^T_n\left({\bm{\rho}}_{k,l}{\bm{\rho}}_{k,l}^T-{\bf{I}}_3\right){\bf{p}}_{k,l}^r}e^{-j\frac{2\pi}{\lambda}\left({\bf{t}}_{m,n}\right)^T{\bm{\rho}}_{k,l}}}.
\end{aligned}
\end{equation}
\end{proposition}
\begin{proof}
See Appendix A.
\end{proof}

\noindent Denote $\frac{\partial{\bf{h}}_k^H({\bf{c}},{\bf{r}})}{\partial{\bf{r}}_n}=[\frac{\partial{\bf{h}}_k^H({\bf{c}},{\bf{r}})}{\partial{{\alpha}}_n},\frac{\partial{\bf{h}}_k^H({\bf{c}},{\bf{r}})}{\partial{{\beta}}_n},\frac{\partial{\bf{h}}_k^H({\bf{c}},{\bf{r}})}{\partial{{\gamma}}_n}]^T$. The expressions of partial derivative $\frac{\partial{\bf{h}}_k^H({\bf{c}},{\bf{r}})}{\partial{\bf{c}}_n}$ and $\frac{\partial{\bf{h}}_k^H({\bf{c}},{\bf{r}})}{\partial{{\Gamma}}_n}$ are shown at the bottom of this page, where ${\Gamma}_n\in\{\alpha_n,\beta_n,\gamma_n\}$.

\begin{algorithm}[t!]
    \renewcommand{\algorithmicrequire}{\bf{Input:}}
    \renewcommand{\algorithmicensure}{\bf{Output:}}
    \caption{Proposed Algorithm for Solving Problem $\eqref{optimization_problem}$}\label{algorithm_2}
    \label{algorithm_1}
    \begin{algorithmic}[1]
    \REQUIRE${\bf{W}}_{\rm{D}}^{(0)},{\bf{W}}_{\rm{A}}^{(0)},{\bf{c}}^{(0)},{\bf{r}}^{(0)},{\bf{u}}^{(0)},{\bf{v}}^{(0)},\{R_k^{(0)}\}_{k=1}^K,\hat{\kappa}_{\bf{c}},\hat{\kappa}_{\bf{r}},$\\
    $~~\quad\varepsilon,I_{max}$.
    \ENSURE ${{\bf{W}}_{\rm{D}}^{*},{\bf{W}}_{\rm{A}}^{*},{\bf{c}}^{*},{\bf{r}}^{*}}$.
    \STATE Initialize $j=1$.
    \REPEAT
    \STATE Obtain ${\bf{u}}^{(j)}$ and ${\bf{v}}^{(j)}$ via \eqref{gamma} and \eqref{omega}, respectively.\
    \STATE Obtain ${\bf{W}}_{\rm{D}}^{(j)}$ via $\eqref{D_optimal}$.
    \STATE Obtain ${\bar{\bf{b}}}^{(j)}$ via algorithm \ref{MO} and recover it to ${\bf{W}}_{\rm{A}}^{(j)}$ based on \eqref{unit_modulus_constraint}.

    \FORALL{$n=1:N$}
    \STATE Calculate the gradient $\nabla_{{\bf{c}}_{n}^{(j)}}{\mathcal{L}}$ and $\nabla_{{\bf{r}}_{n}^{(j)}}{\mathcal{L}}$ via $\eqref{gradient_value}$ and \eqref{gradient_value_r}, respectively.
    \STATE Initialize the step size $\kappa_{\bf{c}}=\hat{\kappa}_{\bf{c}}$ and $\kappa_{\bf{r}}=\hat{\kappa}_{\bf{r}}$.
    \REPEAT
    {\STATE Compute $\hat{{\bf{c}}}_{n}={{\bf{c}}}_{n}^{(j)}+\kappa_{\bf{c}}\nabla_{({\bf{c}}_{n}^{(j)})}{\mathcal{L}}$.
    \STATE Shrink the step size $\kappa_{\bf{c}}\leftarrow\frac{\kappa_{\bf{c}}}{2}$.
    \STATE Update ${{\bf{c}}}_{n}^{\left(j+1\right)}$ according to \eqref{Position_Update}.}
    \UNTIL{$\hat{{\bf{c}}}_{n}\in\mathcal{C}_{n}$ and $\mathcal{L}(\hat{{\bf{c}}}_{n},{{\bf{r}}}_{n}^{(j)})\geq\mathcal{L}({{\bf{c}}}_{n}^{(j)},{{\bf{r}}}_{n}^{(j)})$}.

    \REPEAT
    {\STATE Compute $\hat{{\bf{r}}}_{n}={{\bf{r}}}_{n}^{(j)}+\kappa_{\bf{r}}\nabla_{({\bf{r}}_{n}^{(j)})}{\mathcal{L}}$.
    \STATE Shrink the step size $\kappa_{\bf{r}}\leftarrow\frac{\kappa_{\bf{r}}}{2}$.
    \STATE Update ${{\bf{r}}}_{n}^{\left(j+1\right)}$ according to \eqref{orientation selection}.}
    \UNTIL{$\hat{{\bf{r}}}_{n}\in\mathcal{R}$ and $\mathcal{L}({{\bf{c}}}_{n}^{(j+1)},\hat{{\bf{r}}}_{n})\geq\mathcal{L}({{\bf{c}}}_{n}^{(j+1)},{{\bf{r}}}_{n}^{(j)})$}.
    
    \ENDFOR
    \STATE Calculate objective value $\sum_{k=1}^KR_k^{(j+1)}$ via $\eqref{rate}$.
    \STATE Update $j=j+1$.
    \UNTIL{$\left|\sum_{k=1}^K\left(R_k^{(j+1)}-R_k^{(j)}\right)\right|<\varepsilon$} or the maximum iteration number $I_{max}$ is reached.
    \end{algorithmic}
\end{algorithm}

In the $j$-th iteration, the central point and orientation of the $n$-th UPA are both updated by searching along the gradient directions. Then, we check whether the updated points remain within feasible regions and result in an increase in the objective value. Define $\mathcal{L}({{\bf{c}}}_{n})\triangleq\mathcal{L}({{\bf{c}}}_{n},{\bf{r}}_n^{(j)})$ and $\mathcal{L}({{\bf{r}}}_{n})\triangleq\mathcal{L}({\bf{c}}_n^{(j+1)},{{\bf{r}}}_{n})$, the update rules are given by

$\setcounter{equation}{35}$\begin{equation}\label{Position_Update}
{\bf{c}}_{n}^{\left(j+1\right)}=\left\{
\begin{aligned}
\hat{{\bf{c}}}_{n},&~{\rm{if}}~\hat{{\bf{c}}}_{n}\in\mathcal{C}_{n} ~ {\rm{and}}~ \mathcal{L}\left(\hat{{\bf{c}}}_{n}\right)\geq\mathcal{L}({{\bf{c}}}_{n}^{(j)}),\\
{\bf{c}}_{n}^{\left(j\right)},&~{\rm{otherwise}},
\end{aligned}
\right.
\end{equation}
and
\begin{equation}\label{orientation selection}
{\bf{r}}_{n}^{\left(j+1\right)}=\left\{
\begin{aligned}
\hat{{\bf{r}}}_{n},&~{\rm{if}}~\hat{{\bf{r}}}_{n}\in\mathcal{R} ~ {\rm{and}}~ \mathcal{L}\left(\hat{{\bf{r}}}_{n}\right)\geq\mathcal{L}({\bf{r}}_n^{(j)}),\\
{\bf{r}}_{n}^{\left(j\right)},&~{\rm{otherwise}},
\end{aligned}
\right.
\end{equation}
where $\hat{{\bf{c}}}_{n}={\bf{c}}_{n}^{\left(j\right)}+\kappa_{\bf{c}}\nabla_{{\bf{c}}_{n}}{\mathcal{L}({\bf{c}}_{n}^{(j)}})$ and $\hat{{\bf{r}}}_{n}={\bf{r}}_{n}^{\left(j\right)}+\kappa_{\bf{r}}\nabla_{{\bf{r}}_{n}}{\mathcal{L}({\bf{r}}_{n}^{(j)}})$ with $\kappa_{\bf{c}}$ and $\kappa_{\bf{r}}$ being the step sizes. To be noted, the effectiveness of gradient descent is highly sensitive to the step size. Therefore, before updating each UPA's position, we set out $\kappa_{\bf{c}}$ to a large positive value and then repeatedly update it using $\kappa_{\bf{c}}\leftarrow{\frac{\kappa_{\bf{c}}}{2}}$ until a feasible solution is found such that $\mathcal{L}(\hat{{\bf{c}}}_{n})\geq\mathcal{L}({{\bf{c}}}_{n}^{(j)})$. Then, with given ${\bf{c}}_n^{(j+1)}$, the orientation of the $n$-th UPA is updated following the same procedure.

Algorithm \ref{algorithm_2} outlines the workflow for solving problem \eqref{optimization_problem}. Convergence is confirmed due to the non-decreasing nature of the objective function throughout the iterations and the existence of a finite optimal value. The computational complexities of solving the variables $\{{\bf{W}}_{\rm{D}},{\bf{W}}_{\rm{A}},{\bf{c}},{\bf{r}},{\bf{u}},{\bf{v}}\}$ are scaled with $\mathcal{O}(KN^3)$, $\mathcal{O}(I_{MO}K^2MN)$, $\mathcal{O} (I_{GD}(MN\sum_{k=1}^K{L_k}))$, $\mathcal{O} (I_{GD}(MN\sum_{k=1}^K{L_k}))$,  $\mathcal{O}(KMN)$, and $\mathcal{O}(KMN)$, respectively. Here, $I_{MO}$ and $I_{GD}$ are the number of iterations for Algorithm \ref{algorithm_1} and the gradient descent method, respectively. Furthermore, we denote the number of iterations needed for the convergence of the AO algorithm as $I_{AO}$. Considering the dominant computational steps, the maximum computational complexity of Algorithm \ref{algorithm_2} can be expressed as $\mathcal{O}(I_{AO}(KN^3+I_{MO}K^2MN+I_{GD}(MN\sum_{k=1}^K{L_k})))$.

\section{Numerical Results}

This section presents numerical results to evaluate the effectiveness of our proposed \ac{6dma}-aided hybrid beamforming scheme. Unless stated otherwise, we analyze a scenario in which a \ac{6dma} BS, equipped with $16$ antennas and operating at a carrier frequency of $f_c = 30$ GHz, serves $K=4$ users. The number of $2\times2$ UPAs is $N=4$. The central point of each UPA remains in the $x$-$o$-$z$ coordinate, wherein antenna element spacing is $\frac{\lambda}{2}$, where $\lambda$ represents the wavelength. We assume the number of transmit and receive channel paths for each user is equal, i.e., $L_k^t=L_k^r=L=6$. The distance from the BS to the $k$-th user follows uniformly distributed $\mathcal{U}[20,100]$. The PRM is defined as a diagonal matrix $\Sigma_k={\rm{diag}}([\sigma_{k,1},\sigma_{k,2},\cdots,\sigma_{k,L}])$, where $\sigma_{k,l}\sim\mathcal{CN}(0,(\rho_0d_k^{-\alpha})^2/L)$ with $\rho_0=-40$ dB and $\alpha=2.8$. It is assumed that the elevation and azimuth angles follow the joint probability density function $f(\theta_{k,l},\phi_{k,l})=\frac{\cos\theta_{k,l}}{2\pi}$\cite{MA1}. The movable range for each UPA is $D\times0\times D$ and the rotatable range for the roll, pitch, and yaw angles are all set as $[-\zeta,\zeta]$. In addition, $P=10$ dBm, $\zeta=20^\circ$, $\sigma_k^2=-80$ dBm, $D=2\lambda$, $\hat{\kappa}_{\bf{c}}=\hat{\kappa}_{\bf{r}}=10$, $\varepsilon=10^{-3}$, and $I_{max}=200$. All UPAs initially have orientations set to ${\bf{r}}_n=[0^\circ,0^\circ,0^\circ]$. Each UPA is initialized at the center of its movable region.

The following baseline scenarios are considered to verify the effectiveness of the proposed scheme:
\begin{itemize}
\item {\bf{Fully-digital-FA, fully-digital-MA with flexible position, and fully-digital-MA with flexible orientation}}: The BS utilizes FAs/MAs with flexible position/MAs with flexible orientation with each antenna connected to a separate RF chain, which is capable of independently processing the input signal. By substituting the hybrid beamforming outlined in Section III-A and Section III-B with a closed-form solution of a fully-digital beamformer \cite{MA4}, the sum rate under the three architectures can be obtained by the proposed framework.

\item {\bf{Fully-connected-FA}}: The BS is equipped with FAs with a fully-connected structure \cite{HBF1,HBF2}. The ABF and DBF are derived by approximating the fully-digital beamformer obtained through the fully-digital-FA scheme.

\item {\bf{Sub-connected-FA}}: The BS utilizes FAs with a sub-connected structure \cite{HBF1,HBF2}. The sum rate of the scheme is obtained according to Section III without the position and orientation designs.

\item {\bf{Unpolarized \ac{6dma}}}: The proposed \acp{6dma} structure is employed, optimizing the DBF, ABF, and the positions and orientations of the UPAs based on a channel model without polarization. These parameters are then applied to the polarized channel model to calculate the sum rate.
\end{itemize}

Fig. \ref{Iteration} evaluates the convergence of proposed algorithm \ref{algorithm_2} with different power budget $P$. As observed, the sum rate obtained by algorithm \ref{algorithm_2} increases over iterations until reaching a constant, which validates our analysis in Section III-E. 

\begin{figure}[!t]
    \centering
    \includegraphics[width=0.4\textwidth]{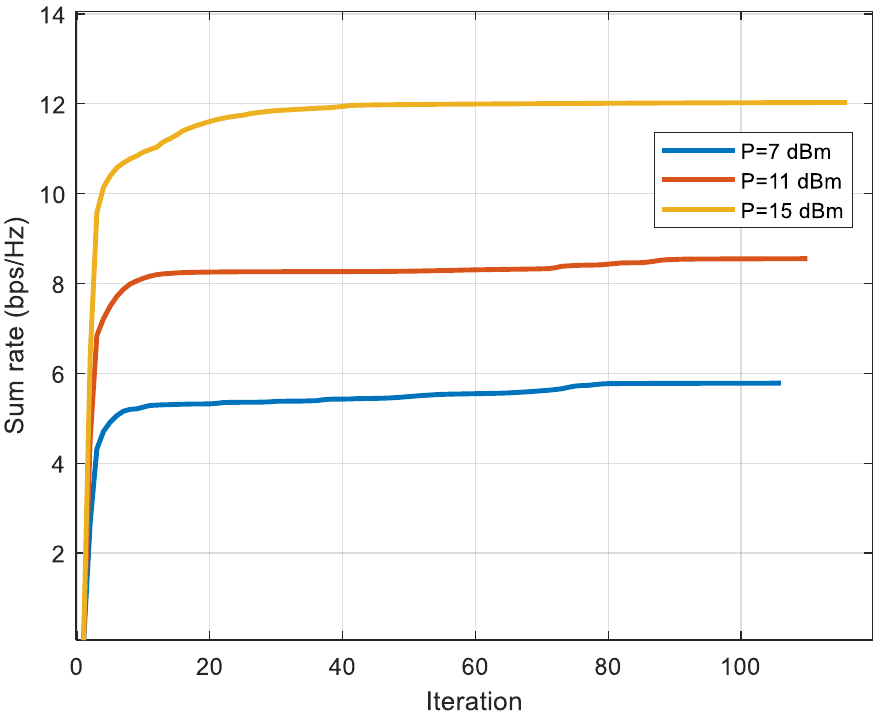}
    \caption{Convergence of Algorithm \ref{algorithm_2}.}
    \label{Iteration}
\end{figure}

\begin{figure}[!t]
    \centering
    \includegraphics[width=0.4\textwidth]{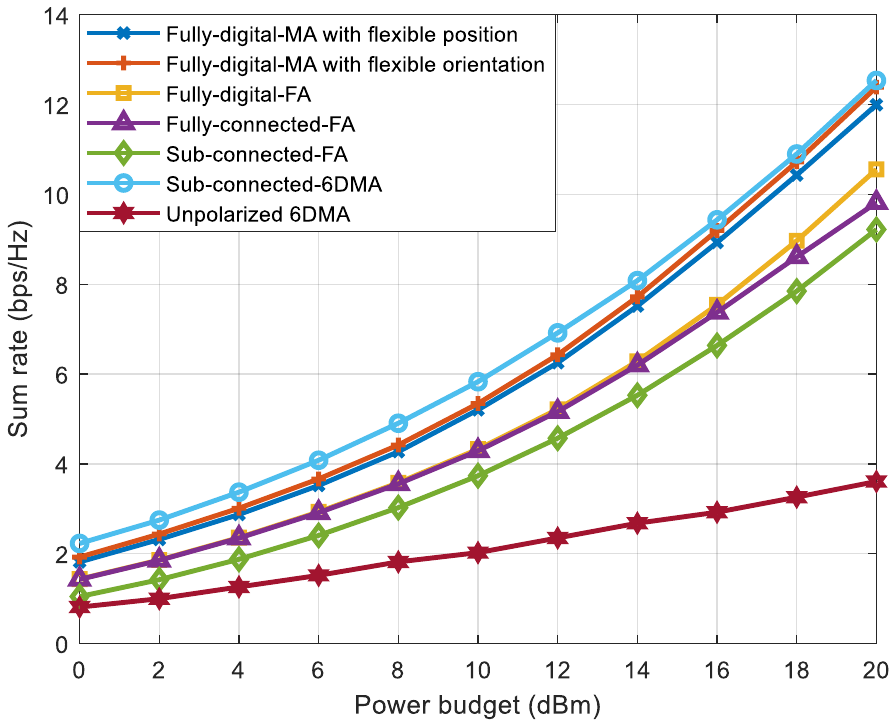}
    \caption{Sum rate versus the power budget.}
    \label{Power_hybrid}
\end{figure}

Fig. \ref{Power_hybrid} depicts the sum rate versus the power budget for different schemes. As expected, the sum rate among all considered schemes increases as the power budget increases. It is observed that for the same power budget $P = 20~ {\rm{dBm}}$, our proposed \ac{6dma} scheme achieves $135\%$, $127\%$, $118\%$, $4.05\%$, and $1.19\%$ performance improvement over the sub-connected-FA, fully-connected-FA, fully-digital-FA, fully-digital-MA with flexible position, and fully-digital-MA with flexible orientation schemes, respectively. This suggests that the performance gain of \acp{6dma} can compensate for the performance loss of using a sub-connected structure. Moreover, the proposed scheme can even outperform fully-digital-MA with flexible position/orientation at sufficiently low transmit power. This indicates that with lower inter-user interference, the \acp{6dma}' ability to enhance received signal power gain significantly contributes to sum rate enhancement. Notably, the performance gap between the unpolarized and proposed \ac{6dma} schemes is substantial, with the \ac{6dma} scheme showing an improvement of about $352\%$ at $P=20~ {\rm{dBm}}$. This highlights the significant performance loss due to polarization mismatch, which implies the importance of constructing models of EM waves' polarization directions \cite{PC_Heath2}. Fortunately, utilizing the \ac{6dma} enables more efficient polarization matching.

\begin{figure}[!t]
    \centering
    \includegraphics[width=0.4\textwidth]{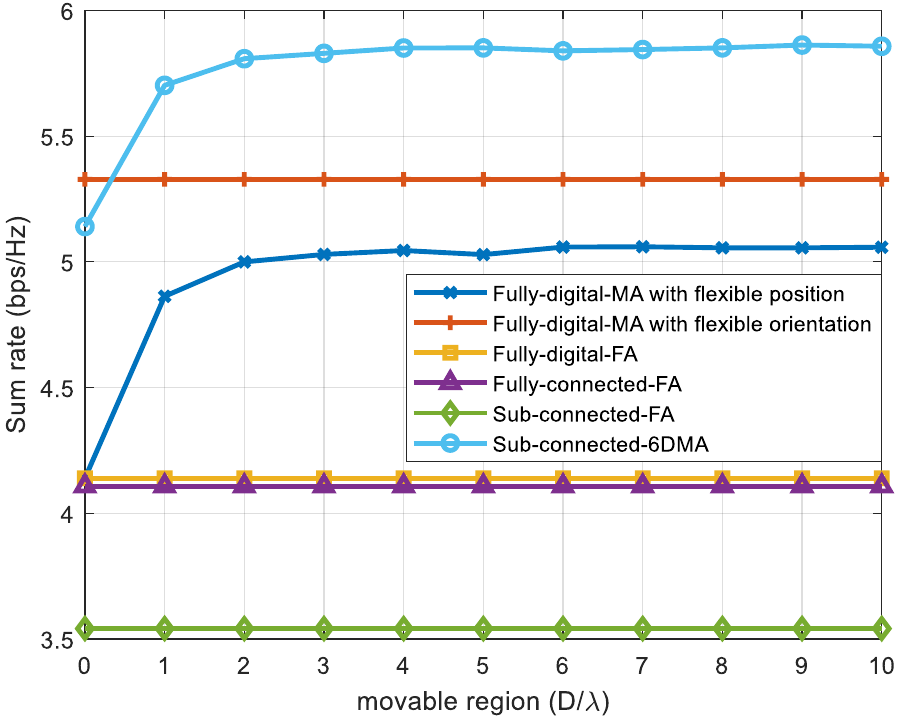}
    \caption{Sum rate versus the size of movable region.}
    \label{Movable_region}
\end{figure}

Fig. \ref{Movable_region} shows the sum rate versus the (normalized) size of movable region $D/\lambda$. As observed, the sum rate of the \ac{6dma} system increases with the increasing size of the movable region until reaching a constant. Moreover, with sufficient large movable region, the \ac{6dma} scheme outperforms all other schemes. The reason behind is that increasing region size provides more DoFs for the \acp{6dma} to exploit the channel spatial variation. However, due to the evident periodicity of the channel gain for each user in the spatial domain, which indicates a strong spatial correlation \cite{MA1}, the optimal channel condition within a finite movable region can be achieved for the sum-rate optimization of the \ac{6dma}-aided system. Therefore, further expansion of the movable region does not lead to additional spatial DoFs. Besides, the performance gap between the sub-connected-MA with flexible orientation (sub-connected-\ac{6dma} without movement) and the fully-digital-MA with flexible orientation is significantly smaller than that observed between the sub-connected-FA and the fully-digital-FA, from $16.8\%$ to $3.64\%$. This implies that, due to the directional beam and polarization effect, the array's rotational capability allows it to reconfigure the channel conditions to better align with the designed beamformer. This adaptability helps mitigate the performance gap between fully-digital and sub-connected structures.

\begin{figure}[!t]
    \centering
    \includegraphics[width=0.4\textwidth]{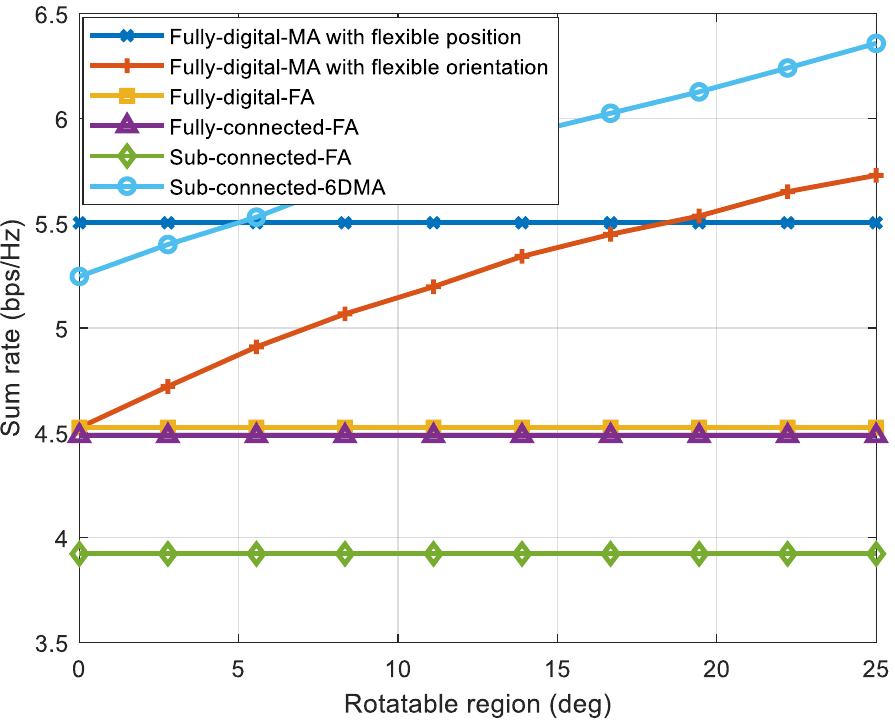}
    \caption{Sum rate versus the range of rotatable angle.}
    \label{Rotatable_region}
\end{figure}

Fig. \ref{Rotatable_region} demonstrates the sum rate versus the range of rotatable angles. The result reveals that the rotational capability of each UPA facilitates the adjustment of the radiation pattern direction and the polarization direction, as well as the specific positions of each antenna within this UPA. Consequently, as the range of the rotatable angles increases, more spatial DoFs can be explored to enhance the sum rate of \ac{6dma}/MA with flexible orientation system. Nevertheless, since the rotation constraints are simplified by a limited range of rotatable angles in simulations, the upper bound for sum rate cannot be achieved. Moreover, Fig. \ref{Rotatable_region} also shows that with sufficiently wide rotatable angles, the \ac{6dma} scheme achieves superior performance among all schemes. Moreover, the sub-connected-MA with flexible position scheme (sub-connected-\ac{6dma} without rotation) achieves a $169\%$ performance improvement over the fully-digital-FA scheme. This indicates that the antenna's movability alone can effectively compensate for the performance loss associated with adopting sub-connected structures.

\begin{figure}[!t]
    \centering
    \includegraphics[width=0.4\textwidth]{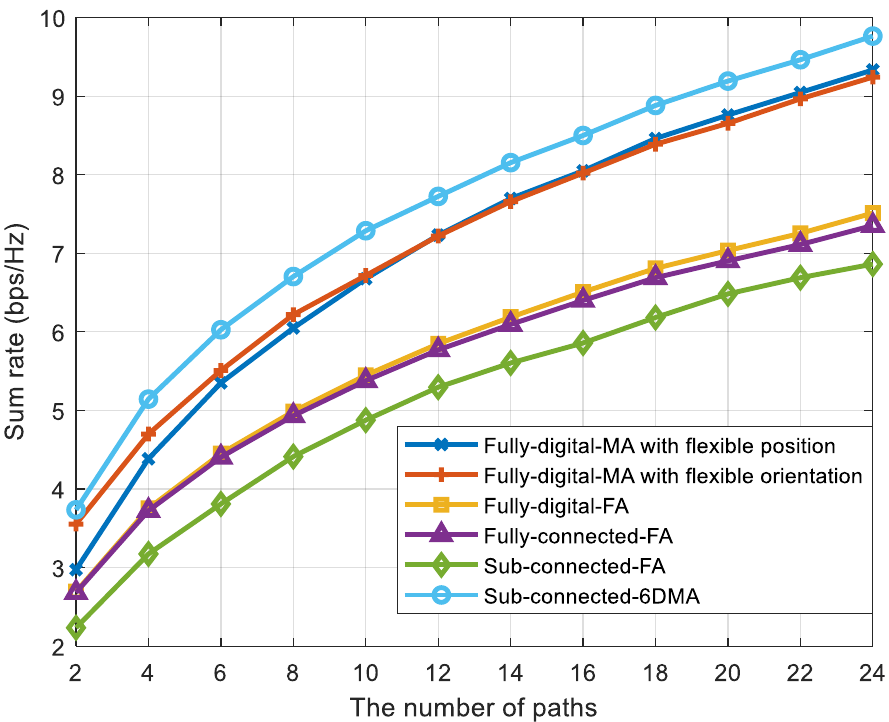}
    \caption{Sum rate versus the number of paths for each user.}
    \label{Path}
\end{figure}

\begin{figure}[!t]
    \centering
    \includegraphics[width=0.4\textwidth]{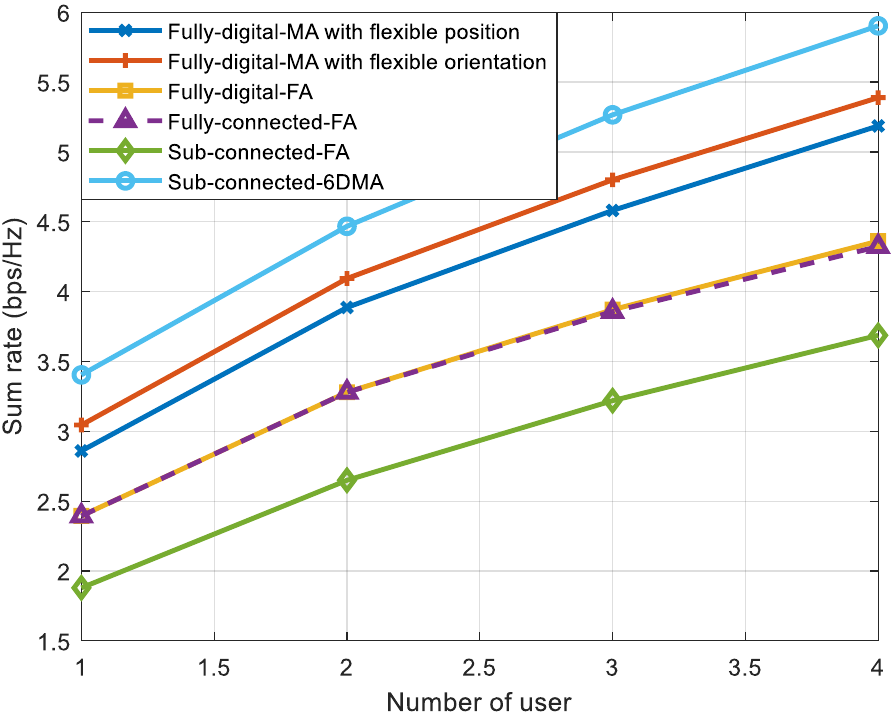}
    \caption{Sum rate versus the number of users.}
    \label{User}
\end{figure}

Fig. \ref{Path} evaluates the sum rate versus the number of paths for each user. We can observe again that the sum rate of the \ac{6dma} system is much higher than all other schemes. The reason for this is the channel gain of the \ac{6dma} system is essentially obtained from the small-scale fading of the multi-path channel. For a large number of paths with equal average power, the strong small-scale fading improves the spatial diversity in the transmit region, which provides more spatial DoFs for the \ac{6dma} system to exploit. Besides, the sum rates of all schemes increase with the number of paths, indicating that the ability to rotate, the mobility, and the beamformer design can all utilize the increased spatial DoFs to mitigate the multi-user interference and enhance the channel gain.

Fig. \ref{User} illustrates the sum rate versus the number of users. Note that the increase in sum rate is primarily attributed to the mitigation of multi-user interference and the improvement of the channel gain for each user. The proposed sub-connected-\ac{6dma} scheme consistently outperforms the five baselines for different numbers of users, which verifies that the \ac{6dma} system with its flexible positions and orientations is capable to create balanced channels for multi-user scenarios. This capability has the potential to reduce the correlation between channel gains, thereby enhancing channel gain for each user while mitigating interference among multiple users.
 
\section{Conclusion}
This paper investigated \ac{6dma}-aided multi-user hybrid beamforming with the sub-connected structure under a practical 6DMA channel model by incorporating directional antenna radiation patterns and antenna polarization. We studied the sum rate maximization problem by jointly designing the DBF, ABF, positions and orientations of UPAs. To solve this non-convex optimization problem, an FP-aided AO framework is developed to transform it into a more tractable form and then decomposed it into four subproblems. These subproblems can be sequentially solved by Lagrange multiplier method, MO and gradient descent. Simulations showed the superiority of the proposed sub-connected-\ac{6dma} scheme compared to the FA and MA based benchmark schemes with less flexibility. Most interestingly, under specific practical conditions, the sub-connected-\ac{6dma} scheme even outperforms the fully-digital schemes with less flexibility in antenna movement. {{The results also revealed that the baseline approach, which optimizes DBF, ABF, and UPA configurations based on an unpolarized channel model, incurs notable performance degradation. This thus necessitates the essential role of incorporating antenna polarization effects in optimizing antenna orientations for \ac{6dma}.}}

\section*{Appendix A\\Proof of Proposition \ref{expressions of channel gains}}
According to \eqref{theta_LCS}, \eqref{phi_LCS}, the definition $\tilde{\theta}_{k,l,n}\in[-\frac{\pi}{2},\frac{\pi}{2}]$ and $\tilde{\phi}_{k,l,n}\in[0,2\pi)$, and the fact that ${\bf{R}}_n^{-1}={\bf{R}}_n^T$, we have

\begin{subequations}\label{sin&cos}
    \begin{align}
        \sin\left(\tilde{\theta}_{k,l,n}\right) &= {\bf{e}}_3^T{\bf{R}}^T_n{\bm{\rho}}_{k,l},\label{sin_theta}\\
       \cos\left(\tilde{\theta}_{k,l,n}\right) &= \sqrt{{1-\left({\bf{e}}_3^T{\bf{R}}^T_n{\bm{\rho}}_{k,l}\right)^2}},\label{cos_theta}\\
       \sin\left(\tilde{\phi}_{k,l,n}\right) &= \frac{{\bf{e}}_2^T{\bf{R}}^T_n{\bm{\rho}}_{k,l}}{\sqrt{\left({\bf{e}}_1^T{\bf{R}}^T_n{\bm{\rho}}_{k,l}\right)^2+\left({\bf{e}}_2^T{\bf{R}}^T_n{\bm{\rho}}_{k,l}\right)^2}},\label{sin_phi}\\
       \cos\left(\tilde{\phi}_{k,l,n}\right) &= \frac{{\bf{e}}_1^T{\bf{R}}^T_n{\bm{\rho}}_{k,l}}{\sqrt{\left({\bf{e}}_1^T{\bf{R}}^T_n{\bm{\rho}}_{k,l}\right)^2+\left({\bf{e}}_2^T{\bf{R}}^T_n{\bm{\rho}}_{k,l}\right)^2}}.\label{cos_phi}
    \end{align}
\end{subequations}
By substituting \eqref{cos_theta} into \eqref{antenna aperture}, we have

\begin{equation}\label{e_solution}
\begin{aligned}
    A_E\left(\tilde{\theta}_{k,l,n}\right)&=\sqrt{\frac{3}{2}}\left|\cos\left(\tilde{\theta}_{k,l,n}\right)\right|\\
    &=\sqrt{\frac{3}{2}\left({1-\left({\bf{e}}_3^T{\bf{R}}^T_n{\bm{\rho}}_{k,l}\right)^2}\right)}.
\end{aligned}
\end{equation}

Next, combining \eqref{sin&cos}, \eqref{spherical unit vector}, and the fact that $({\bf{e}}_1^T{\bf{R}}^T_n{\bm{\rho}}_{k,l})^2+({\bf{e}}_2^T{\bf{R}}^T_n{\bm{\rho}}_{k,l})^2+({\bf{e}}_3^T{\bf{R}}^T_n{\bm{\rho}}_{k,l})^2=1$, we have

\begin{equation}
    {\bf{e}}_{\tilde{{\theta}}_{k,l,n}}=\begin{bmatrix}
        \frac{{\bf{e}}_1^T{\bf{R}}^T_n{\bm{\rho}}_{k,l}{\bf{e}}_3^T{\bf{R}}^T_n{\bm{\rho}}_{k,l}}{\sqrt{{1-\left({\bf{e}}_3^T{\bf{R}}^T_n{\bm{\rho}}_{k,l}\right)^2}}}\\\frac{{\bf{e}}_2^T{\bf{R}}^T_n{\bm{\rho}}_{k,l}{\bf{e}}_3^T{\bf{R}}^T_n{\bm{\rho}}_{k,l}}{\sqrt{{1-\left({\bf{e}}_3^T{\bf{R}}^T_n{\bm{\rho}}_{k,l}\right)^2}}}\\-\sqrt{{1-\left({\bf{e}}_3^T{\bf{R}}^T_n{\bm{\rho}}_{k,l}\right)^2}}
    \end{bmatrix}.
\end{equation}
Then, the polarization vector can be derived as

\begin{equation}\label{polarization_vector_t}
    \begin{aligned}
        {\bf{p}}_{k,l,n}^{t} &\stackrel{\left(a\right)} = \left({\bf{e}}_{{\theta}_{k,l}}{\bf{e}}_{{\theta}_{k,l}}^T+{\bf{e}}_{{\phi}_{k,l}}{\bf{e}}_{{\phi}_{k,l}}^T\right){\bf{R}}_n\left[{\bf{e}}_1,{\bf{e}}_2,{\bf{e}}_3\right]{\bf{e}}_{\tilde{{{\theta}}}_{k,l,n}}\\
        &\stackrel{\left(b\right)} = \left({\bf{I}}_3-{\bm{\rho}}_{k,l}{\bm{\rho}}_{k,l}^T\right)\frac{\left({\bm{\rho}}_{k,l}{\bm{\rho}}_{k,l}^T-{\bf{I}}_3\right){\bf{R}}_n{\bf{e}}_3}{\sqrt{{1-\left({\bf{e}}_3^T{\bf{R}}^T_n{\bm{\rho}}_{k,l}\right)^2}}}\\
        &=\frac{\left({\bm{\rho}}_{k,l}{\bm{\rho}}_{k,l}^T-{\bf{I}}_3\right){\bf{R}}_n{\bf{e}}_3}{\sqrt{{1-\left({\bf{e}}_3^T{\bf{R}}^T_n{\bm{\rho}}_{k,l}\right)^2}}},
    \end{aligned}
\end{equation}
where $(a)$ holds due to ${\bf{I}}_3\triangleq\left[{\bf{e}}_1,{\bf{e}}_2,{\bf{e}}_3\right]$ and $(b)$ holds because ${\bm{\rho}}_{k,l}$, ${\bf{e}}_{{\theta}_{k,l}}$, and ${\bf{e}}_{{\phi}_{k,l}}$ are a set of orthogonal basis vectors such that ${\bm{\rho}}_{k,l}{\bm{\rho}}_{k,l}^T+{\bf{e}}_{{\theta}_{k,l}}{\bf{e}}_{{\theta}_{k,l}}^T+{\bf{e}}_{{\phi}_{k,l}}{\bf{e}}_{{\phi}_{k,l}}^T={\bf{I}}_3$. By substituting \eqref{polarization_vector_t} into \eqref{polarization matching}, the channel gain from polarization matching can be derived as

\begin{equation}\label{p_solution}
    A_P\left(\tilde{\theta}_{k,l,n},\tilde{\phi}_{k,l,n}\right)=\frac{{\bf{e}}_3^T{\bf{R}}^T_n\left({\bm{\rho}}_{k,l}{\bm{\rho}}_{k,l}^T-{\bf{I}}_3\right){\bf{p}}_{k,l}^r}{\sqrt{{1-\left({\bf{e}}_3^T{\bf{R}}^T_n{\bm{\rho}}_{k,l}\right)^2}}}.
\end{equation}
The channel between the $m$-th antenna within the $n$-th UPA and the $k$-th user is calculated directly by substituting \eqref{e_solution} and \eqref{p_solution} into \eqref{general channel model}. This completes the proof of Proposition \ref{expressions of channel gains}.

\bibliographystyle{IEEEtran}
\bibliography{main.bib}
\end{document}